\documentclass[a4paper,UKenglish]{lipics-v2016}
 
\usepackage{microtype}

\usepackage{amsmath}
\usepackage[all=normal,paragraphs=tight,bibliography=tight]{savetrees}
\usepackage{amsthm}
\usepackage{amssymb}
\usepackage{microtype}
\usepackage{xspace}
\usepackage{comment}
\usepackage{caption}
\usepackage{subcaption}
\usepackage{letltxmacro}
\usepackage{comment}
\usepackage{enumerate}
\usepackage{tikz}
\usepackage{eufrak}
\usepackage{xcolor}
\usepackage{colortbl}
\usepackage{graphicx}
\usepackage{scalefnt}
\usepackage{authblk}
\usepackage{multirow}
\usetikzlibrary{fit,matrix}

\usepackage{todonotes}

\usepackage{thm-restate}

\usepackage{boxedminipage}

\usepackage{bold-extra}
\usepackage{lmodern}
\usepackage[T1]{fontenc}
\usepackage{textcomp}

\newcommand{\PSI}{{\sc Partitioned Subgraph Isomorphism}}

\newcommand{\Z}{{\mathbb Z}_{\scriptscriptstyle{\geq 0}}}
\newcommand{\R}{{\mathbb R}}
\newcommand{\ZZ}[1]{\mathbb{Z}_{#1}}
\newcommand{\BB}{{\mathcal B}}
\newcommand{\CC}{{\mathcal C}}
\newcommand{\IP}{(IP)\xspace}
\newcommand{\bw}{{branch-width}\xspace}
\newcommand{\pw}{{path-width}\xspace}
\usepackage{graphicx}
\usepackage{caption}
\usepackage{subcaption}

\newcommand{\cmm}{A_{\psi}}
\newcommand{\tvm}{b_{\psi}}

\newtheorem{claim}[theorem]{Claim}
 \newtheorem{proposition}[theorem]{Proposition}

\newcommand{\OO}{{\mathcal{O}}}
\newcommand{\ii}{i}
\newcommand{\ip}{i'}
\newcommand{\q}{q}

\newcommand{\bb}{g}

\usepackage{thm-restate}

\usepackage{booktabs}

\newcommand{\defproblem}[3]{
  \vspace{1mm}
\noindent\fbox{
  \begin{minipage}{0.96\textwidth}
  #1 \\
  {\bf{Input:}} #2  \\
  {\bf{Question:}} #3
  \end{minipage}
  }
}

\title{On the Optimality of Pseudo-polynomial Algorithms for Integer Programming}

\author[1]{Fedor V.~Fomin}
\author[1]{Fahad Panolan}
\author[2]{M. S. Ramanujan}
\author[1,3]{Saket Saurabh}
\affil[1]{University of Bergen, Norway. \texttt{\{fomin|fahad.panolan\}@ii.uib.no}}
\affil[2]{University of Warwick. \texttt{R.Maadapuzhi-Sridharan@warwick.ac.uk}}
\affil[3]{Institute of Mathematical Sciences, HBNI, Chennai, India. \texttt{saket@imsc.res.in}}

\begin{document}

\maketitle


\begin{abstract}
In the classic  
{\em Integer Programming} (IP) problem,  the objective is to decide whether, for a given $m \times n$ matrix $A$ and an $m$-vector $b=(b_1,\dots, b_m)$, there is a non-negative integer $n$-vector $x$ such that $Ax=b$. 
Solving (IP) is an important step in numerous algorithms and it is important to obtain an understanding of the precise complexity of this problem as a function of natural parameters of the input.

The classic pseudo-polynomial time algorithm of  Papadimitriou~[J. ACM 1981] for instances of (IP) with a constant number of constraints was only recently improved upon by Eisenbrand and Weismantel~[SODA 2018] and Jansen and Rohwedder~[ArXiv 2018]. We continue this line of work and  show that under  the Exponential Time Hypothesis (ETH), the algorithm of Jansen and Rohwedder is nearly optimal. We also show that when the matrix $A$ is assumed to be non-negative, a component of Papadimitriou's original algorithm is already nearly optimal under ETH.

This motivates us to pick up the line of research initiated by Cunningham and Geelen~[IPCO 2007] who studied the complexity of solving (IP) with non-negative matrices in which the number of constraints may be unbounded, but the  branch-width of the column-matroid corresponding to the constraint matrix is a constant. We prove a lower bound on the complexity of solving (IP) for such instances and obtain optimal results with respect to a closely related parameter, path-width. Specifically, we 
prove {\em matching} upper and lower bounds for   (IP) when the {\em path-width} of the corresponding column-matroid is a constant. 

 \end{abstract}

\section{Introduction}\label{sec:intro}

\newcommand{\cm}{A_{(\psi,c)}}
\newcommand{\tv}{b_{(\psi,c)}}
\newcommand{\valueA}{\Delta}
\newcommand{\valueAB}{d}
\newcommand{\valueB}{\|b\|_{\infty}}

In the classic {\em Integer Programming} problem, the input is an  $m\times n$ integer matrix $A$, and an $m$-vector $b=(b_1, \dots, b_m)$. 
We consider the feasibility version of the problem, where 
the objective is to find a non-negative integer $n$-vector  $x $ (if one exists)  such that  $Ax=b$. Solving this problem, denoted by \IP, is a fundamental step in numerous algorithms and it is important to obtain an understanding of the precise complexity of this problem as a function of natural parameters of the input.

\IP is known to be NP-hard. However,  there are two classic algorithms  due to  Lenstra \cite{ILP:Lenstra} and 
Papadimitriou ~\cite{Papadimitriou81} solving \IP in polynomial or pseudo-polynomial time for two important cases when the number of variables and the number of constraints are bounded. These algorithms  in some sense complement each other. 

The algorithm of Lenstra 
shows that \IP is solvable in polynomial time when the number of variables is bounded. Actually, the result of Lenstra is even stronger:   \IP is 
{\em fixed-parameter tractable} parameterized by the number of variables. 
 However, the running time of Lenstra's algorithm is  
doubly exponential in  $n$. Later,  
Kannan \cite{ILP:Kannan} provided an 
algorithm for  \IP  running in time $n^{\OO(n)}$.  Deciding whether the running time  $n^{\OO(n)}$ can be improved to  $2^{\OO(n)}$  is a long-standing open question. 

Our work is motivated by the complexity analysis of the complementary case when the number of constraints is bounded.  \IP is NP-hard already for $m=1$ (the \textsc{Knapsack} problem) but solvable in pseudo-polynomial time.  
In 1981, Papadimitriou~\cite{Papadimitriou81}   extended this result by showing  that \IP 
is solvable in  pseudo-polynomial time on instances for which the number of constraints $m$  is a constant. 
 The  algorithm of Papadimitriou consists of two steps. The first step is combinatorial, showing that if the entries of $A$ and $b$ are from $\{0,\pm 1,  \dots, \pm \valueAB\},$ and \IP has a solution,  then there is also a solution 
which is in $\{0,1,  \dots,  n (m\valueAB )^{2m+1}\}^n$. 
The second, algorithmic step shows that if \IP has a solution  with the maximum entry at most $B$, then the problem is solvable in time $\OO((nB)^{m+1})$.
Thus the total running time of Papadimitriou's algorithm is 
$\OO(n^{2m+2} \cdot (m\valueAB)^{(m+1)(2m+1)})$, where  $\valueAB$ is an upper bound on the absolute values of the entries of $A$ and $b$. There was no algorithmic progress on this problem until the very recent breakthrough of   Eisenbrand and Weismantel~\cite{EisenbrandW18}. They proved the following result. 

\begin{proposition}[Theorem 2.2, \textbf{Eisenbrand and Weismantel~\cite{EisenbrandW18}}]\label{prop:Eisen}
\IP with    $m\times n$ matrix $A$  is solvable  in time $(m\cdot \valueA)^{\OO(m)}\cdot \valueB^2$, where $\valueA$ is an  upper bound on the absolute values of the entries of $A$.

\end{proposition}

%
%

Then, Jansen and Rohwedder improved Proposition~\ref{prop:Eisen} and gave a matching lower bound very recently~\cite{JansenR18}. 

\begin{proposition}[\textbf{Jansen and Rohwedder~\cite{JansenR18}}]\label{prop:Jansen}
\IP with    $m\times n$ matrix $A$  is solvable  in time 
$\OO(m \valueA)^m \log(\valueA) \log(\valueA+\valueB)$. 
where $\valueA$ is an  upper bound on the absolute values of the entries of $A$. 
Assuming the Strong Exponential Time Hypothesis (SETH), 
there is no algorithm for \IP\ running in time $n^{\OO(1)}\cdot \OO(m(\valueA+\valueB))^{m-\delta}$ for any $\delta>0$. 
\end{proposition}

SETH is the hypothesis that CNF-SAT cannot be solved in time $(2-\epsilon)^nm^{\OO(1)}$ on $n$-variable $m$-clause formulas for any constant $\epsilon$. ETH is the hypothesis that 3-SAT cannot be solved in time $2^{o(n)}$ on $n$-variable formulas. 
Both ETH and SETH were first introduced in the work of Impagliazzo and Paturi~\cite{ImpagliazzoP01}, which built upon earlier work of Impagliazzo, Paturi and Zane~\cite{ImpagliazzoPZ01}. 
 One of the natural question is whether the exponential dependence of $\valueB$ can be improved 
 significantly   
 at the cost of super polynomial dependence on $n$. 
  Our first theorem provides a conditional lower bound indicating that any significant  improvements are unlikely. 


\begin{restatable}{theorem}{ETHIP}\label{thm:ETHIP}
 Unless the Exponential Time Hypothesis (ETH)  fails, 
\IP with $m\times n$  matrix $A$ cannot be solved in time $n^{o(\frac{m}{\log m})} \cdot \valueB^{o(m)}$ 
  even when the 
constraint matrix $A$  is non-negative and each entry in any feasible solution 
is at most $2$. 
\end{restatable}

 Let us note that  since the bound in  Theorem~\ref{thm:ETHIP} holds for a non-negative matrix $A$, we can always reduce (in polynomial time) the original instance of the problem to an equivalent instance   
where the maximum value $\valueA$ in the 
constraint matrix $A$   does not exceed $\valueB$. Thus  Theorem~\ref{thm:ETHIP} also implies  the conditional lower bound
$n^{o(\frac{m}{\log m})} \cdot (\valueA\cdot \valueB)^{o(m)}$. When $m=\OO( n)$, our bound also implies the lower bound
$(n\cdot m)^{o(\frac{m}{\log m})} \cdot (\valueA\cdot \valueB)^{o(m)}$.
 We complement  Theorem~\ref{thm:ETHIP} by 
turning our focus to the dependence of algorithms solving \IP on $m$ alone, and obtaining the following theorem.

\begin{restatable}{theorem}{ETHIPtwo}\label{thm:ETHIP2}
 Unless the Exponential Time Hypothesis (ETH)  fails,  
\IP with $m\times n$  matrix $A$ cannot be solved in time $f(m)\cdot  (n \cdot \valueB)^{o(\frac{m}{\log m})}$ 
 for any computable function $f$. The result holds even when the 
constraint matrix $A$  is non-negative and each entry in any feasible solution 
is at most $1$. 
\end{restatable}

 The difference between our first two theorems is the following. 
Although   Theorem~\ref{thm:ETHIP} provides a better dependence on $\valueB$,   Theorem~\ref{thm:ETHIP2} provides much more information on how the complexity of the problem depends on $m$. 
   Since several parameters are involved in this running time estimation, a natural objective is to study the possible tradeoffs between them. For instance, consider the 
 $\OO(m \valueA)^m \log(\valueA) \log(\valueA+\valueB)$ time 
 algorithm (Proposition~\ref{prop:Jansen}) for \IP. A natural follow up question is the following. 
  Could it be that by allowing a significantly worse  dependence (a superpolynomial dependence) on $n$ and $\valueB$  and an {\em arbitrary} dependence on $m$, one might be able to improve the dependence on $\valueA$ alone?  Theorem~\ref{thm:ETHIP2} provides a strong argument against such an eventuality. 
 Indeed, since the lower bound of 
  Theorem~\ref{thm:ETHIP2} holds even for non-negative matrices, it rules out algorithms with running time $f(m)\cdot  \valueA^{o(\frac{m}{\log m})} \cdot (n \cdot \valueB)^{o(\frac{m}{\log m})}$. Therefore, obtaining a subexponential dependence of 
  $\valueA$ on ${m}$ even at the cost of a superpolynomial dependence of $n$ and $\valueB$ on $m$, and an arbitrarily bad dependence on $m$ is as hard as obtaining a subexponential algorithm for 3-SAT.

We now motivate our remaining results. We refer the reader to Figure~\ref{fig:summary} for a summary of our main results. 
It is straightforward to see that when the matrix $A$ happens to be non-negative, the  algorithm of Papadimitriou~\cite{Papadimitriou81}  runs in time $\OO((n\cdot \valueB)^{m+1})$. 
Due to Theorems~\ref{thm:ETHIP} and ~\ref{thm:ETHIP2}, the  dynamic programming step of the algorithm of Papadimitriou
for \IP{} when the maximum entry in a solution as well as in the constraint matrix  is bounded, is  already close to optimal.  %
%
%
%
Consequently, any quest for  ``faster'' algorithms for \IP must be built around the use of additional structural properties of the matrix $A$.
Cunningham and Geelen~\cite{CunninghamG07} introduced such an approach by considering  the \emph{branch decomposition} of the matrix $A$. They were motivated by the fact that the result of Papadimitriou can be interpreted as a result for matrices of constant \emph{rank} and {\bw} is a parameter which is upper bounded by rank plus one.  
%
 For a matrix $A$, the \emph{column-matroid} of $A$ denotes the matroid whose elements are the columns of $A$ and whose independent sets are precisely the linearly independent sets of columns of $A$.  
We postpone the formal definitions of branch decomposition and \bw  till the next section. For \IP with a \emph{non-negative} matrix $A$,  Cunningham and Geelen \cite{CunninghamG07} showed that when the \bw of the column-matroid of $A$ is constant, \IP is solvable in pseudo-polynomial time. 
\begin{proposition}[\textbf{Cunningham and Geelen \cite{CunninghamG07}}]\label{thmCG}
\IP with { a non-negative}  $m\times n$ matrix $A$ given together with a branch decomposition of its column matroid of width $k$, is solvable  in time $\OO((\valueB+1)^{2k}mn + m^2n)$.
\end{proposition}

We analyze the complexity of \IP parameterized by the \bw of $A$, by making use  of SETH 
 and 
 obtain the following lower bound(s).

\begin{theorem}
\label{thm:lowbranchwidth0}
 Unless SETH fails, \IP with  a non-negative $m\times n$  constraint matrix $A$ cannot be solved in time 
 $f({\sf bw})(\valueB+1)^{(1-\epsilon){\sf bw}}(mn)^{\OO(1)}$ or $f(\valueB)(\valueB+1)^{(1-\epsilon){\sf bw}}(mn)^{\OO(1)}$,  for any computable function $f$. Here $\sf bw$ is the branchwidth of the column matroid of $A$.
\end{theorem}

In recent years, 
 SETH has been used to obtain several tight conditional bounds on the running time of algorithms for various optimization problems on graphs of bounded treewidth 
~\cite{LokshtanovMS11a-tw}. However, in order to be able to use SETH to prove lower bounds for \IP in combination with the \bw of matroids, we have to develop new ideas. 

In fact, Theorem~\ref{thm:lowbranchwidth0} follows from  stronger lower bounds we prove using the \pw of $A$ as our parameter of interest instead of the \bw.
  The parameter \emph{\pw} is closely related to the notion of 
\emph{trellis-width} of a linear code, which is a parameter commonly used in coding theory
 \cite{DBLP:journals/tit/HornK96}.
For a matrix $A\in  \R^{m\times n}$, computing the \pw of the column matroid of $A$ is equivalent to computing  the  trellis-width of the linear code generated by $A$. Roughly speaking,  
the \pw of  the column matroid of 
$A$ is at most $k$, if there is a permutation of  the columns of $A$ such that in the matrix $A'$ obtained from $A$ by applying this column-permutation, for every $1\leq i \leq n-1$, the \emph{dimension} of the subspace of $\R^{m}$ obtained by taking the intersection of the subspace of $\R^{m}$ spanned by the first $i$ columns with the subspace of $\R^{m}$ spanned by the remaining columns, is at most $k-1$. 
%

\begin{figure}[t]
\centering
\scriptsize
{
\begin{tabular}{|c|c|c|c|c|c|}
\hline
Upper Bounds  &  Lower bounds   \\ \hline \hline 
 &   \\ 
 
& no $n^{o(\frac{m}{\log m})} \cdot \valueB^{o(m)}$ time algorithm under ETH (Theorem \ref{thm:ETHIP})  \\  
 & (\emph{even} for non-negative matrix $A$ and solution entries bounded by 2)  \\ 
  
 &   \\  \hline 
 \\
 
 \multirow{2}{*}{
 $(m\cdot \valueA)^{\OO(m)}\cdot \valueB^{\OO(1)}$~\cite{EisenbrandW18,JansenR18}} & no $n^{\OO(1)}\cdot \OO(m(\valueA+\valueB))^{m-\delta}$ time algorithm for $\delta>0$ under SETH~\cite{JansenR18} \vspace{4pt} \\ \cline{2-2} 
& no $(n\cdot m)^{o(\frac{m}{\log m})}  (\valueA\cdot \valueB)^{o(m)}$  algorithm when $m=\OO(n)$  under ETH \vspace{3pt} \\ & (consequence of Theorem \ref{thm:ETHIP})  \\     
 & (\emph{even} for non-negative matrix $A$ and solution entries bounded by 2)  \\ 
  
 &   \\  \hline 
 
%

\\
& no $f(m)\cdot  (n \cdot \valueB)^{o(\frac{m}{\log m})}$ under ETH (Theorem \ref{thm:ETHIP2})  \\  
 & (\emph{even} for non-negative matrix $A$ and solution entries bounded by 1)  \\ 
  
 &   \\  \hline

 &   \\  
$\OO((\valueB +1)^{{\sf pw}+1}mn + m^2n)$   & no $f({\sf pw})(\valueB+1)^{(1-\epsilon){\sf pw}}(mn)^{\OO(1)}$ algorithm under SETH (Theorem \ref{thm:lowbranchwidth}) \\  
 (non-negative matrix $A$)
 &  (\emph{even} for non-negative matrix $A$) \\  (Theorem \ref{thmCGlin})
 &   \\  
  &   no $f(\valueB)(\valueB+1)^{(1-\epsilon){\sf pw}}(mn)^{\OO(1)}$ algorithm under SETH (Theorem \ref{thm:lowentries}) \\ 
  &  (\emph{even} for non-negative matrix $A$)  \\  \\ \hline 
     
     \\
$\OO((\valueB +1)^{2{\sf bw}}mn + m^2n)$  & no $f({\sf bw})(\valueB+1)^{(1-\epsilon){\sf bw}}(mn)^{\OO(1)}$ \\ (non-negative matrix $A$)~\cite{CunninghamG07}& or \\ & $f(\valueB)(\valueB+1)^{(1-\epsilon){\sf bw}}(mn)^{\OO(1)}$ algorithm \\ & under SETH (Theorem~\ref{thm:lowbranchwidth0})  \\  
 & (\emph{even} for non-negative matrix $A$)  \\ 
  
 &   \\  \hline

\end{tabular}
}
\caption{ 
A summary of our lower bound results in comparison with the relevant known upper bound results. Here, $n$ and $m$ are the number of variables and constraints respectively, ${\sf pw}$ and ${\sf bw}$ denote the {\pw} and {\bw} of the column matroid of $A$ and $\valueB$ 
denotes a bound on the   largest absolute value in $b$ while $\valueA$ denotes a bound on the largest absolute value in $A$.
} \label{fig:summary} 
\end{figure}


The value of the parameter {\pw} is always at least the value of {\bw} and thus Theorem~\ref{thm:lowbranchwidth0} follows from the following theorems.

\begin{theorem}
\label{thm:lowbranchwidth}
 Unless SETH fails, \IP with even a non-negative $m\times n$  constraint matrix $A$ cannot be solved in time $f(k)(\valueB+1)^{(1-\epsilon)k}(mn)^{\OO(1)}$ for any computable function 
$f$ and $\epsilon>0$,  where 
  $k$ is the \pw of the column matroid of $A$.
\end{theorem}

\begin{theorem}
\label{thm:lowentries} 
 Unless SETH fails, \IP with even a non-negative $m\times n$  constraint matrix $A$ cannot be solved in time $f(\valueB)(\valueB+1)^{(1-\epsilon)k}(mn)^{\OO(1)}$ for any computable function 
$f$ and $\epsilon>0$,  where 
 $k$ is the \pw of the column matroid of $A$.
\end{theorem}


 
Although the proofs of both lower bounds have a similar {structure}, we believe that there are sufficiently many differences in the proofs to warrant stating and proving them separately.

Note that although there is still a gap between  the  upper bound of Cunningham and Geelen from Proposition~\ref{thmCG}
and the lower bound provided by
Theorem~\ref{thm:lowbranchwidth0}, the lower bounds given in Theorems~\ref{thm:lowentries} and \ref{thm:lowbranchwidth} are asymptotically tight in the following sense. The proof of  
  Cunningham and Geelen in  \cite{CunninghamG07}
    actually implies the upper bound stated in Theorem~\ref{thmCGlin}. We   provide a self-contained proof in this paper for the  reader's convenience. 
  
\begin{theorem}\label{thmCGlin}
\IP with non-negative  $m\times n$ matrix $A$ given together with a path decomposition of its column matroid of width $k$ is solvable  in time $\OO((\valueB+1)^{k+1}mn + m^2n)$.
\end{theorem}
Then by  Theorem~\ref{thm:lowbranchwidth},  we cannot relax the $(\valueB+1)^k$ factor in Theorem~\ref{thmCGlin} {even} if we allow in the running time  an arbitrary function depending on $k$,  while Theorem~\ref{thm:lowentries} shows a similar lower bound in terms of $\valueB$ instead of $k$. Put together the results imply that no matter how much one is allowed to compromise on either the path-width or the bound on $\valueB$, 
 it is unlikely that the algorithm of Theorem \ref{thmCGlin} can be improved.

%
%

The \pw of matrix $A$ does not exceed its rank and thus the number of constraints in \IP. Hence, similar to  Proposition~\ref{thmCG}, Theorem~\ref{thmCGlin} generalizes the result of Papadimitriou when restricted to non-negative matrices. 
Also we note that   the assumption of non-negativity is unavoidable (without any further assumptions such as a bounded domain for the variables) in this setting because \IP is NP-hard  when the constraint matrix $A$ is allowed to have  negative values (in fact even when restricted to $\{-1,0,1\}$) and the branchwidth of the column matroid of $A$ is at most 3. A close inspection of the instances they construct in their NP-hardness reduction shows that the column matroids of the resulting constraint matrices are in fact direct sums of circuits, implying that even their \emph{\pw} is bounded by 3.

\medskip
\noindent
{\bf Organization of the paper.}
The rest of the paper is organized as follows. There are two main technical parts to this paper. The first part (Section~\ref{Sec:ETHlowerbounds}) is devoted to proving Theorem~\ref{thm:ETHIP} and Theorem~\ref{thm:ETHIP2} (our ETH based lower bounds) while the second part (Section~\ref{sec:SETHLB}) is devoted to proving Theorem~\ref{thm:lowbranchwidth} and Theorem~\ref{thm:lowentries} (our SETH based lower bounds), and consequently, Theorem~\ref{thm:lowbranchwidth0}. 
For all our reductions, we begin by giving an overview in order to help the reader (especially in the SETH based reductions)  navigate the technical details in the reductions.  
%
We then prove Theorem~\ref{thm:lowentries} in  Section~\ref{sec:lowentries} and  Theorem~\ref{thmCGlin} in Section~\ref{sec:thmCGlin} (completing the results for constant {\pw}).

\section{Preliminaries}\label{sec:prel}

We assume that the reader is familiar with basic definitions 
from linear algebra, matroid theory and graph theory.  

\noindent\textbf{Notations.} 
We use $\Z$ and $\R$ to denote the set of non negative integers and real numbers, respectively. For any positive 
integer $n$, we use $[n]$ and $\ZZ{n}$ to denotes the sets 
$\{1,\ldots,n\}$ and $\{0,1,\ldots,n-1\}$, respectively. 
For convenience, we say that $[0]=\emptyset$. 
For any two vectors $b, b'\in \R^m$ and $i\in [m]$, we use $b[i]$ to denote 
the $i^{th}$ coordinate of $b$ and we write $b'\leq b$, if $b'[i]\leq b[i]$ for all $i\in [m]$. 
We often use $0$ to denote the zero-vector whose length will be clear from the context.  
For a matrix $A\in  \R^{m\times n}$, $I\subseteq [m]$ and $J\subseteq [n]$, 
$A[I,J]$ denote the submatrix of $A$ obtained by the restriction of $A$ to the 
rows indexed by $I$ and columns indexed by $J$.   
For an $m\times n$ matrix $A$ and $n$-vector $v$, 
we can write $Av=\sum_{i=1}^{n} A_i v[i]$, 
where $A_i$ is the $i^{th}$ column of $A$. Here we say that 
$v[i]$ is a multiplier of column $A_i$.  For convenience, in this paper, we consider $0$ as an even number. 

\medskip\noindent\textbf{Branch-width of matroids.}  The notion of the branch-width of  graphs, and 
implicitly of matroids, was introduced by Robertson and Seymour in \cite{RobertsonS91}. 
Let  ${M} = (U, {\cal F})$  be a matroid with universe  set $U$  and family ${\cal F}$ of independent sets over $U$.  
We use $r_M$ to denote the rank function of $M$. 
That is, for any $S\subseteq U$, $r_M(S)=\max_{S'\subseteq S, S'\in {\cal F}} \vert S'\vert$.   
 For $X\subseteq U$, the \emph{connectivity function} 
  of $M$ is defined as   \[\lambda_M(X)=r_M(X) +r_M(U\setminus X) - r_M(U) +1\]

For matrix $A\in  \R^{m\times n}$, we use $M(A)$ to denote the column-matroid of $A$.  In this case the connectivity function $\lambda_{M(A)}$ has the following interpretation. For $E=\{1,\dots, n\}$ and $X\subseteq E$, we define 
\[S(A,X)=\operatorname{span}(A|X)\cap\operatorname{span}(A|E\setminus X), 
\]
where $A|X$ is the set of columns of $A$ restricted to $X$ and $\operatorname{span}(A|X)$ is the   subspace of $\R^{m}$ 
spanned by the columns $A|X$. It is easy to see that the dimension of $S(A,X)$ is equal to 
$\lambda_{M(A)}(X)-1$.

A tree is \emph{cubic} if its internal vertices all have degree $3$. 
A \emph{branch decomposition} of 
matroid ${M}$ with universe  set $U$ is  a cubic  tree $T$ and mapping $\mu$ which maps elements of 
$U$ to leaves of $T$. Let  $e$ be an edge of $T$. Then   the forest $T-e$ consists of two connected 
components $T_1$ and $T_2$. Thus every edge $e$ of $T$ corresponds to the partitioning of $U$ into two sets
 $X_e$ and $U\setminus X_e$ such that 
$\mu(X_e)$ are the leaves of $T_1$ and  $\mu(U\setminus X_e)$ are the leaves of $T_2$. The \emph{width} of edge $e$
is $\lambda_{M}(X_e)$ and the width of branch decomposition  $(T, \mu)$ is the maximum edge width, where maximum is taken over all edges of $T$. Finally, the \emph{branch-width} of $M$ is the minimum width taken over all possible branch decompositions of $M$.  

 
 The \emph{\pw} of a matroid is defined as follows. Recall that a \emph{caterpillar} is a tree which is obtained from a path by attaching leaves to some vertices of the path. Then the \pw of a matroid is   the minimum width of a branch decomposition $(T,\mu)$, where $T$ is a cubic caterpillar.  Let us note that every mapping of elements of a matroid to the leaves of a cubic caterpillar corresponds to an  ordering of these elements. 
 Jeong, Kim, and  Oum \cite{Jeong0O16} gave a constructive fixed-parameter tractable algorithm to construct a path decomposition of width at most $k$ for a column matroid of a given matrix.
 
 \medskip\noindent\textbf{ETH and SETH.}
 For $q\geq 3$, let $\delta_q$ be the infimum of the set of constants $c$ for which there exists an algorithm solving $q$-SAT with $n$ variables and $m$ clauses in time $2^{cn}\cdot m^{\OO(1)}$.  
The {\em{Exponential-Time Hypothesis} (ETH)} and {\em{Strong Exponential-Time Hypothesis} (SETH)} are then formally defined as follows.
ETH conjectures  that $\delta_3>0$ and  SETH that $\lim_{q\to \infty}\delta_q=1$.



\section{ ETH lower bounds on pseudopolynomial solvability of \IP} \label{Sec:ETHlowerbounds}

In this section we prove  Theorems~\ref{thm:ETHIP} and~\ref{thm:ETHIP2}. 

\subsection{Proof of Theorem~\ref{thm:ETHIP}}\label{sec:lowrank}
This subsection is devoted to the proof of Theorem~\ref{thm:ETHIP}


\ETHIP*

%

Our proof is  by  a reduction from {\sc $3$-CNF SAT} to \IP. 
There are exactly 2 variables in the \IP\ instance for each variable (one for each literal) and clause. For each clause we define two constraints. For each variable in the 3-CNF formula, we have a constraint, which is a selection gadget. 
%

\begin{figure}
\centering

\begin{tabular}{lllllllllllllll}
                        & ${x_1}$                                                     & ${\bar x_1}$                                                  & ${x_2}$                                                     & ${\bar x_2}$                                                  & ${x_3}$                                                     & ${\bar x_3}$                                                 & ${x_4}$                                                     & ${\bar x_4}$                                                 & ${Y_1}$                                                     & ${Z_1}$                                                   & ${Y_2}$                                                     & ${Z_2}$                                                    & ${Y_3}$                                                    & ${Z_3}$                                                    \\ \cline{2-15}

\multicolumn{1}{l|}{$C_1$} & \multicolumn{1}{l|}{{\color[HTML]{000000} \textbf{1}}} & \multicolumn{1}{l|}{{\color[HTML]{000000} \textbf{0}}} & \multicolumn{1}{l|}{{\color[HTML]{000000} \textbf{1}}} & \multicolumn{1}{l|}{{\color[HTML]{000000} \textbf{0}}} & \multicolumn{1}{l|}{{\color[HTML]{000000} \textbf{1}}} & \multicolumn{1}{l|}{{\color[HTML]{000000} \textbf{0}}} & \multicolumn{1}{l|}{{\color[HTML]{000000} \textbf{0}}} & \multicolumn{1}{l|}{{\color[HTML]{000000} \textbf{0}}} & \multicolumn{1}{l|}{{\color[HTML]{000000} \textbf{1}}} & \multicolumn{1}{l|}{{\color[HTML]{000000} }}           & \multicolumn{1}{l|}{{\color[HTML]{000000} }}           & \multicolumn{1}{l|}{{\color[HTML]{333333} }}           & \multicolumn{1}{l|}{{\color[HTML]{333333} }}           & \multicolumn{1}{l|}{{\color[HTML]{333333} }}           \\ \cline{2-15} 

\multicolumn{1}{l|}{}   & \multicolumn{1}{l|}{{\color[HTML]{000000} \textbf{}}}  & \multicolumn{1}{l|}{{\color[HTML]{000000} \textbf{}}}  & \multicolumn{1}{l|}{{\color[HTML]{000000} \textbf{}}}  & \multicolumn{1}{l|}{{\color[HTML]{000000} \textbf{}}}  & \multicolumn{1}{l|}{{\color[HTML]{000000} \textbf{}}}  & \multicolumn{1}{l|}{{\color[HTML]{000000} \textbf{}}}  & \multicolumn{1}{l|}{{\color[HTML]{000000} \textbf{}}}  & \multicolumn{1}{l|}{{\color[HTML]{000000} \textbf{}}}  & \multicolumn{1}{l|}{{\color[HTML]{000000} \textbf{1}}} & \multicolumn{1}{l|}{{\color[HTML]{000000} \textbf{1}}} & \multicolumn{1}{l|}{{\color[HTML]{000000} \textbf{}}}  & \multicolumn{1}{l|}{{\color[HTML]{333333} \textbf{}}}  & \multicolumn{1}{l|}{{\color[HTML]{333333} \textbf{}}}  & \multicolumn{1}{l|}{{\color[HTML]{333333} \textbf{}}}  \\ \cline{2-15}

\multicolumn{1}{l|}{$C_2$} & \multicolumn{1}{l|}{{\color[HTML]{000000} \textbf{0}}} & \multicolumn{1}{l|}{{\color[HTML]{000000} \textbf{1}}} & \multicolumn{1}{l|}{{\color[HTML]{000000} \textbf{0}}} & \multicolumn{1}{l|}{{\color[HTML]{000000} \textbf{1}}} & \multicolumn{1}{l|}{{\color[HTML]{000000} \textbf{1}}} & \multicolumn{1}{l|}{{\color[HTML]{000000} \textbf{0}}} & \multicolumn{1}{l|}{{\color[HTML]{000000} \textbf{0}}} & \multicolumn{1}{l|}{{\color[HTML]{000000} \textbf{0}}} & \multicolumn{1}{l|}{{\color[HTML]{000000} }}           & \multicolumn{1}{l|}{{\color[HTML]{000000} }}           & \multicolumn{1}{l|}{{\color[HTML]{000000} \textbf{1}}} & \multicolumn{1}{l|}{{\color[HTML]{333333} \textbf{}}}  & \multicolumn{1}{l|}{{\color[HTML]{333333} \textbf{}}}  & \multicolumn{1}{l|}{{\color[HTML]{333333} \textbf{}}}  \\ \cline{2-15} 

\multicolumn{1}{l|}{}   & \multicolumn{1}{l|}{{\color[HTML]{000000} \textbf{}}}  & \multicolumn{1}{l|}{{\color[HTML]{000000} \textbf{}}}  & \multicolumn{1}{l|}{{\color[HTML]{000000} \textbf{}}}  & \multicolumn{1}{l|}{{\color[HTML]{000000} \textbf{}}}  & \multicolumn{1}{l|}{{\color[HTML]{000000} \textbf{}}}  & \multicolumn{1}{l|}{{\color[HTML]{000000} \textbf{}}}  & \multicolumn{1}{l|}{{\color[HTML]{000000} \textbf{}}}  & \multicolumn{1}{l|}{{\color[HTML]{000000} \textbf{}}}  & \multicolumn{1}{l|}{{\color[HTML]{000000} \textbf{}}}  & \multicolumn{1}{l|}{{\color[HTML]{000000} \textbf{}}}  & \multicolumn{1}{l|}{{\color[HTML]{000000} \textbf{1}}} & \multicolumn{1}{l|}{{\color[HTML]{333333} \textbf{1}}} & \multicolumn{1}{l|}{{\color[HTML]{333333} \textbf{}}}  & \multicolumn{1}{l|}{{\color[HTML]{333333} \textbf{}}}  \\ \cline{2-15}

\multicolumn{1}{l|}{$C_3$}   & \multicolumn{1}{l|}{{\color[HTML]{000000} \textbf{0}}} & \multicolumn{1}{l|}{{\color[HTML]{000000} \textbf{0}}} & \multicolumn{1}{l|}{{\color[HTML]{000000} \textbf{0}}} & \multicolumn{1}{l|}{{\color[HTML]{000000} \textbf{1}}} & \multicolumn{1}{l|}{{\color[HTML]{000000} \textbf{0}}} & \multicolumn{1}{l|}{{\color[HTML]{000000} \textbf{1}}} & \multicolumn{1}{l|}{{\color[HTML]{000000} \textbf{0}}} & \multicolumn{1}{l|}{{\color[HTML]{000000} \textbf{1}}} & \multicolumn{1}{l|}{{\color[HTML]{000000} }}           & \multicolumn{1}{l|}{{\color[HTML]{000000} }}           & \multicolumn{1}{l|}{{\color[HTML]{000000} \textbf{}}}  & \multicolumn{1}{l|}{{\color[HTML]{333333} \textbf{}}}  & \multicolumn{1}{l|}{{\color[HTML]{333333} \textbf{1}}} & \multicolumn{1}{l|}{{\color[HTML]{333333} \textbf{}}}  \\ \cline{2-15}

\multicolumn{1}{l|}{}   & \multicolumn{1}{l|}{{\color[HTML]{000000} \textbf{}}}  & \multicolumn{1}{l|}{{\color[HTML]{000000} \textbf{}}}  & \multicolumn{1}{l|}{{\color[HTML]{000000} \textbf{}}}  & \multicolumn{1}{l|}{{\color[HTML]{000000} \textbf{}}}  & \multicolumn{1}{l|}{{\color[HTML]{000000} \textbf{}}}  & \multicolumn{1}{l|}{{\color[HTML]{000000} \textbf{}}}  & \multicolumn{1}{l|}{{\color[HTML]{000000} \textbf{}}}  & \multicolumn{1}{l|}{{\color[HTML]{000000} \textbf{}}}  & \multicolumn{1}{l|}{{\color[HTML]{000000} \textbf{}}}  & \multicolumn{1}{l|}{{\color[HTML]{000000} \textbf{}}}  & \multicolumn{1}{l|}{{\color[HTML]{000000} \textbf{}}}  & \multicolumn{1}{l|}{{\color[HTML]{333333} \textbf{}}}  & \multicolumn{1}{l|}{{\color[HTML]{333333} \textbf{1}}} & \multicolumn{1}{l|}{{\color[HTML]{333333} \textbf{1}}} \\ \cline{2-15}

\multicolumn{1}{l|}{$x_1$}   & \multicolumn{1}{l|}{{\color[HTML]{000000} \textbf{1}}} & \multicolumn{1}{l|}{{\color[HTML]{000000} \textbf{1}}} & \multicolumn{1}{l|}{{\color[HTML]{000000} \textbf{}}}  & \multicolumn{1}{l|}{{\color[HTML]{000000} \textbf{}}}  & \multicolumn{1}{l|}{{\color[HTML]{000000} \textbf{}}}  & \multicolumn{1}{l|}{{\color[HTML]{000000} \textbf{}}}  & \multicolumn{1}{l|}{{\color[HTML]{000000} \textbf{}}}  & \multicolumn{1}{l|}{{\color[HTML]{000000} \textbf{}}}  & \multicolumn{1}{l|}{{\color[HTML]{000000} \textbf{}}}  & \multicolumn{1}{l|}{{\color[HTML]{000000} \textbf{}}}  & \multicolumn{1}{l|}{{\color[HTML]{000000} \textbf{}}}  & \multicolumn{1}{l|}{{\color[HTML]{333333} \textbf{}}}  & \multicolumn{1}{l|}{{\color[HTML]{333333} \textbf{}}}  & \multicolumn{1}{l|}{{\color[HTML]{333333} \textbf{}}}  \\ \cline{2-15} 
\multicolumn{1}{l|}{$x_2$}   & \multicolumn{1}{l|}{{\color[HTML]{000000} \textbf{}}}  & \multicolumn{1}{l|}{{\color[HTML]{000000} \textbf{}}}  & \multicolumn{1}{l|}{{\color[HTML]{000000} \textbf{1}}} & \multicolumn{1}{l|}{{\color[HTML]{000000} \textbf{1}}}  & \multicolumn{1}{l|}{{\color[HTML]{000000} \textbf{}}}  & \multicolumn{1}{l|}{{\color[HTML]{000000} \textbf{}}}  & \multicolumn{1}{l|}{{\color[HTML]{000000} \textbf{}}}  & \multicolumn{1}{l|}{{\color[HTML]{000000} \textbf{}}}  & \multicolumn{1}{l|}{{\color[HTML]{000000} \textbf{}}}  & \multicolumn{1}{l|}{{\color[HTML]{000000} \textbf{}}}  & \multicolumn{1}{l|}{{\color[HTML]{000000} \textbf{}}}  & \multicolumn{1}{l|}{{\color[HTML]{333333} \textbf{}}}  & \multicolumn{1}{l|}{{\color[HTML]{333333} \textbf{}}}  & \multicolumn{1}{l|}{{\color[HTML]{333333} \textbf{}}}  \\ \cline{2-15} 
\multicolumn{1}{l|}{$x_3$}   & \multicolumn{1}{l|}{{\color[HTML]{000000} \textbf{}}}  & \multicolumn{1}{l|}{{\color[HTML]{000000} \textbf{}}}  & \multicolumn{1}{l|}{{\color[HTML]{000000} \textbf{}}}  & \multicolumn{1}{l|}{{\color[HTML]{000000} \textbf{}}}  & \multicolumn{1}{l|}{{\color[HTML]{000000} \textbf{1}}} & \multicolumn{1}{l|}{{\color[HTML]{000000} \textbf{1}}} & \multicolumn{1}{l|}{{\color[HTML]{000000} \textbf{}}}  & \multicolumn{1}{l|}{{\color[HTML]{000000} \textbf{}}}  & \multicolumn{1}{l|}{{\color[HTML]{000000} \textbf{}}}  & \multicolumn{1}{l|}{{\color[HTML]{000000} \textbf{}}}  & \multicolumn{1}{l|}{{\color[HTML]{000000} \textbf{}}}  & \multicolumn{1}{l|}{{\color[HTML]{333333} \textbf{}}}  & \multicolumn{1}{l|}{{\color[HTML]{333333} \textbf{}}}  & \multicolumn{1}{l|}{{\color[HTML]{333333} \textbf{}}}  \\ \cline{2-15} 
\multicolumn{1}{l|}{$x_4$}   & \multicolumn{1}{l|}{{\color[HTML]{000000} \textbf{}}}  & \multicolumn{1}{l|}{{\color[HTML]{000000} \textbf{}}}  & \multicolumn{1}{l|}{{\color[HTML]{000000} \textbf{}}}  & \multicolumn{1}{l|}{{\color[HTML]{000000} \textbf{}}}  & \multicolumn{1}{l|}{{\color[HTML]{000000} \textbf{}}}  & \multicolumn{1}{l|}{{\color[HTML]{000000} \textbf{}}}  & \multicolumn{1}{l|}{{\color[HTML]{000000} \textbf{1}}} & \multicolumn{1}{l|}{{\color[HTML]{000000} \textbf{1}}} & \multicolumn{1}{l|}{{\color[HTML]{000000} \textbf{}}}  & \multicolumn{1}{l|}{{\color[HTML]{000000} \textbf{}}}  & \multicolumn{1}{l|}{{\color[HTML]{000000} \textbf{}}}  & \multicolumn{1}{l|}{{\color[HTML]{333333} \textbf{}}}  & \multicolumn{1}{l|}{{\color[HTML]{333333} \textbf{}}}  & \multicolumn{1}{l|}{{\color[HTML]{333333} \textbf{}}}  \\ \cline{2-15} 
\end{tabular}

\caption{An illustration of the matrix $A_\psi$ corresponding to the 3-CNF formula $\psi=(x_1\vee x_2\vee x_3) \wedge (\bar x_1 \vee \bar x_2 \vee x_3) \wedge (\bar x_4 \vee \bar x_2 \vee \bar x_3)$. The unfilled cells have 0 as the entry. 
}
\label{fig:Apsi-illustration}
\end{figure}

%

We now proceed to the formal description of the reduction. 
From a {\sc $3$-CNF} formula $\psi$ on $n$ variables and $m$ clauses 
we create an equivalent  \IP{} instance $\cmm x=\tvm,x\geq 0$, where 
$\cmm$ is a non-negative integer  $(2m+n)\times 2 (m+n)$ matrix   and the largest 
entry in $\tvm$ is $3$.   Our reduction can be easily seen to be a  polynomial time reduction and we do not give an explicit analysis.  
Let $\psi$ be the input of {\sc $3$-CNF SAT}. Let $X=\{x_1,\ldots,x_n\}$ be the 
set of variables in $\psi$ and $\CC=\{C_1,\ldots,C_m\}$ be the set of clauses in 
$\psi$. 
First we define the set of variables in the in the \IP\ instance. For each $x_i\in X$, we have two variables 
$x_i$ and $\overline{x}_i$ in the \IP{} instance $\cmm x=\tvm,x\geq 0$. 
For each $C_i\in \CC$, we have two variables $Y_i$ and  $Z_i$. 

Now we define the set of constraints of $\cmm x=\tvm,x\geq 0$.
For each $C_i=x\vee y \vee z$, we define two constraints  
\begin{eqnarray}
x+y+z+Y_i&=&3 \qquad \mbox{and} \label{eqn:ci1} \\
Y_i+Z_i&=&2.\label{eqn:ci2}
\end{eqnarray}
\begin{eqnarray}
\mbox{For each $i\in [n]$,} \qquad\qquad x_i+\overline{x}_i=1 \label{eqn:xi} 
\end{eqnarray}
This completes the construction of \IP{} instance $\cmm x=\tvm,x\geq 0$. See Figure~\ref{fig:Apsi-illustration} for an illustration. 
%
%
%
%
%
%
%
We now argue that this reduction correctly maps satisfiable 3-CNF formulas to feasible instances of \IP and vice versa.

\begin{lemma} 
\label{lemma:rank_correctness}
The formula $\psi$ is satisfiable if and only if $\cmm x=\tvm,x\geq 0$ is feasible.
\end{lemma}
\begin{proof}
Suppose that the formula $\psi$ is satisfiable and let $\phi$ be a satisfying assignment of $\psi$. We set values for 
the variables $\{x_i,\overline{x}_i \colon i\in [n]\}\cup \{Y_i,Z_i \colon i\in [m]\}$ 
and prove that $\cmm x=\tvm$. For any $i\in [n]$, if $\phi(x_i)=1$ we set $x_i=1$ and $\overline{x}_i=0$. Otherwise, we set $x_i=0$ and $\overline{x}_i=1$. 
 
For every $i\in [m]$, we define
\begin{equation}
Y_i = 
\left\{ \begin{array}{ll}
0 & \mbox{if the number of literals set to $1$ in $C_{i}$ by $\phi$ is $3$,}  \\
1 & \mbox{if the number of literals set to $1$ in $C_{i}$ by $\phi$ is $2$,} \\
2 & \mbox{otherwise,}
\end{array}\right. \label{eqn:x:bottom:odd}
\end{equation}
and
\begin{equation}
Z_i = 
\left\{ \begin{array}{ll}
2 & \mbox{if the number of literals set to $1$ in $C_{i}$ by $\phi$ is $3$,}  \\
1 & \mbox{if the number of literals set to $1$ in $C_{i}$ by $\phi$ is $2$,} \\
0 & \mbox{otherwise.}
\end{array}\right. \label{eqn:x:bottom:even}
\end{equation}
We now proceed to prove that the above substitution of values to the variables is indeed a feasible solution. 
Towards this, we need to show that \eqref{eqn:ci1}, \eqref{eqn:ci2}, and \eqref{eqn:xi} are satisfied.  First 
consider \eqref{eqn:ci1}. Let $C_i=x\vee y \vee z$. Since $\phi$ is a satisfying assignment, we have that $1\leq x+y+z \leq 3$. Thus, 
by \eqref{eqn:x:bottom:odd}, we conclude that $x+y+z+Y_i=3$. Because of \eqref{eqn:x:bottom:odd} and \eqref{eqn:x:bottom:even},  \eqref{eqn:ci2} is satisfied. Since the values for $\{x_i,\overline{x}_i\colon i\in [n]\}$ is derived from an assignment $\phi$, \eqref{eqn:xi} 
is satisfied.

For the converse direction of the statement of the lemma, suppose that there exists non-negative values for 
the set of variables $\{x_i,\overline{x}_i \colon i\in [n]\}\cup \{Y_i,Z_i \colon i\in [m]\}$, such that  \eqref{eqn:ci1}, \eqref{eqn:ci2}, and \eqref{eqn:xi} are satisfied.
Now we need to show that 
$\psi$ is satisfiable. Because of \eqref{eqn:xi}, we know that exactly one of $x_i$ and $\overline{x}_i$ is set to one and other is set to zero. 
%
%
Next, we define an assignment $\phi$ 
and prove that $\phi$ is a satisfying assignment for $\psi$.   For  $i\in [n]$ we define
\[
\phi(x_i) = 
\left\{ \begin{array}{ll}
1 & \mbox{if } x_i=1, \\
0 & \mbox{if } \overline{x}_i=1.
\end{array}\right.
\]

We claim that $\phi$ satisfies all the clauses. Consider a clause $C_j=x\vee y \vee z$ where $j\in [m]$. 
Since $Y_j+Z_j=2$ (by \eqref{eqn:ci2}), we have that $Y_i\in \{0,1,2\}$. 
Since $Y_i\in \{0,1,2\}$, by \eqref{eqn:ci1}, at least one of $x,y$ or $z$ is set to one. This implies that 
$\phi$ satisfies $C_j$.  
%
This completes the proof of   the lemma. 
\end{proof} 

By  \eqref{eqn:ci2} and \eqref{eqn:xi}, 
we have that the value set for any variable in a feasible solution is at most $2$. 
The following lemma completes the proof of the theorem. 

\begin{lemma}
If there is an algorithm for 
 \IP{} running  in time $n^{o(\frac{m}{\log m})} \valueB^{o(m)}$,  then 
ETH fails.
\end{lemma}
\begin{proof}
By the Sparsification Lemma~\cite{ImpagliazzoPZ01}, we know that {\sc $3$-CNF SAT} on $n'$ variables and $cn'$ clauses, where $c$ is a constant,  
cannot be solved in time $2^{o(n')}$ time. 
Suppose there is an algorithm {\sf ALG} for  
\IP{} running in time $n^{o(\frac{m}{\log m})} \valueB^{o(m)}$. Then for a $3$-CNF formula $\psi$ with $n'$ 
variables and $m'=cn$ clauses we create an instance $\cmm x=\tvm,$ $x\geq 0$ of \IP{} 
as discussed in this section, in polynomial time, where $\cmm$ is a matrix of dimension 
$(2cn'+n')\times (2(n'+cn'))$ and the largest entry in $\tvm$ is $3$. 
 Then by Lemma~\ref{lemma:rank_correctness}, we can  run  {\sf ALG} to  test whether $\psi$ is satisfiable or not. 
This takes time $$(2(cn'+n'))^{o(\frac{2cn'+n')}{\log (2cn'+n')})} \cdot 3^{o(2cn'+n')}=2^{o(n')},$$ hence refuting ETH. 
\end{proof}

\subsection{Proof of Theorem~\ref{thm:ETHIP2}}\label{sec:lowrank2}
In this section we prove the following theorem. 
 
\ETHIPtwo*


\medskip
Towards proving Theorem~\ref{thm:ETHIP2} we use the ETH based lower bound result of Marx~\cite{marx-toc-treewidth} for {\PSI}. For two graphs $G$ and $H$, a map $\phi\colon V(G)\mapsto V(H)$ is called 
a {\em subgraph isomorphism} from $G$ to $H$, if $\phi$ is injective and for any $\{u,v\}\in E(G)$, $\{\phi(u),\phi(v)\}\in E(H)$ (see Figure~\ref{fig:PSI_example} for an illustration).

\begin{figure}[t]
\begin{center}
  \includegraphics[height=175 pt, width=280 pt]{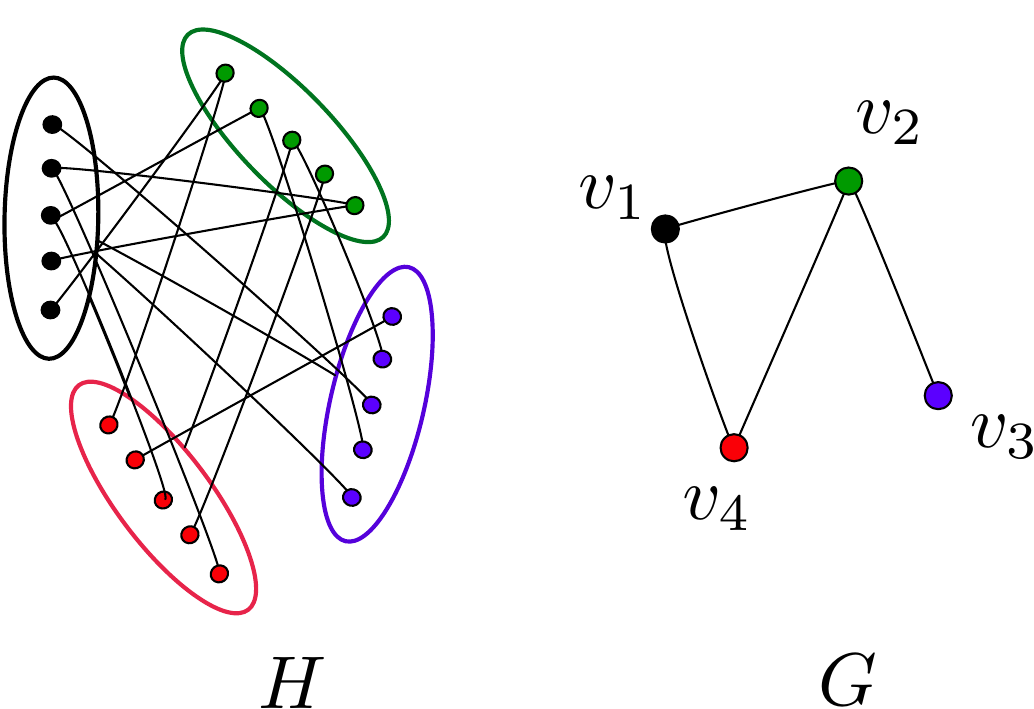}

  \caption{An illustration of an instance of {\PSI}.}  
  \label{fig:PSI_example}
  \end{center}
\end{figure}

\defproblem{\PSI}{Two graphs $G,H$, a bijection $c_G\colon V(G)\mapsto [\ell]$ and a function $c_H\colon V(H)\mapsto [\ell]$, where $\ell=\vert V(G) \vert $.}
{Is there a subgraph isomorphism $\phi$ from $G$ to $H$ such that for any $v\in V(G)$, $c_G(v)=c_H(\phi(v))$?}

\begin{lemma}[Corollary 6.3~\cite{marx-toc-treewidth}]
\label{lem:psimarx}
If \PSI\ can be solved in time $f(G)n^{o(\frac{k}{\log k})}$, where $f$ is
an arbitrary function, $n=\vert V(H)\vert$ and $k$ is the number of edges of the smaller graph $G$, then ETH fails.
\end{lemma}

To prove Theorem~\ref{thm:ETHIP2} we give a polynomial time reduction from \PSI\ to \IP such that for every instance $(G,H,c_G,c_H)$ of \PSI\, the reduction outputs an instance of \IP where the constraint matrix has dimension $\OO(\vert E(G)\vert)\times\OO(\vert E(H)\vert)$ and the largest value in the target vector is $\max \{ \vert E(H)\vert, \vert V(H)\vert\}$.


Let $(G,H,c_G,c_H)$ be an instance of \PSI. Let $k=\vert E(G)\vert$ and $n=\vert V(H)\vert$. We construct an instance $Ax=b$ of \IP from $(G,H,c_G,c_H)$ in polynomial time.  Without loss of 
generality we assume that $[n]=V(H)$ and that there are no isolated vertices in $G$. Hence, the number of vertices in $G$ is at most $2k$. Let $m=\vert E(H)\vert$. For each $e\in E(H)$ we assign a unique integer 
from $[m]$. Let $\alpha\colon E(H)\mapsto [m]$ be the bijection which represents the assignment mentioned above.    
For any $i,j\in [\ell]$, we use $E_H(i,j)$ as a shorthand for the set of edges of $H$ between 
$c_H^{-1}(i)$ and $c_H^{-1}(j)$. Finally, for ease of presentation we let  $\{v_1,\ldots,v_{\ell}\}=V(G)$ and $c_G(v_i)=i$ for all $i\in [\ell]$, where $\ell=\vert V(G)\vert$. 

For illustrative purposes, before proceeding to the formal construction, 
we give an informal description of the \IP\ instance we obtain from a specific  instance of {\PSI}. Let $H$ and $G$ be the graphs in Figure~\ref{fig:PSI_example} and consider the graph $\widehat{H}$ obtained from $H$ as depicted in Figure~\ref{fig:PSI_example_aux_graph}.

For every color $i\in [\ell]$ we have a column in $\widehat{H}$ and for every pair of distinct colors $i,j\in [\ell]$ such that $\{v_i,v_j\}\in E(G)$, we have a copy of $c_H^{-1}(i)$ in Column $i$ and Row $i$ and a copy of $c_H^{-1}(i)$ in Column $i$ and Row $j$. Thus, Column $i$  comprises at most $\ell$ copies of the vertices of $H$ whose image under $c_H$ is $i$ and Row $i$ comprises a copy of $c_H^{-1}(i)$ and additionally, a copy of every vertex $u$ of $H$ such that $v_{c_H{u}}$ is adjacent to $v_i$ in $G$. That is, the color of $u$ is ``adjacent'' to the color $i$ in $G$.

For a vertex $u\in V(H)$, we refer to the unique copy of $u$ in the $i^{th}$ row as the $i^{th}$ copy of $u$ in $\widehat H$.
For every edge $e=\{a,b\}\in E(H)$ where $c_H(a)=i$, $c_H(b)=j$, and $\{v_i,v_j\}\in E(G)$, we have two copies of $e$ in $\widehat H$. The first copy of $e$ has as its endpoints, the $i^{th}$ copy of $a$ and the $i^{th}$ copy of $b$ and the second copy of $e$ has as its endpoints, the $j^{th}$ copy of $a$ and the $j^{th}$ copy of $b$. We now rephrase the {\PSI} problem (informally) as a problem of finding a certain type of subgraph in $\widehat{H}$, which in turn will point us in the direction of our {\IP} instance in a natural way. The rephrased problem statement is the following. Given $G$, $H$,$c_H$,$c_G$ and the resulting auxiliary graph $\widehat H$, find a set of $2|E(H)|$ edges in $\widehat H$ such that the following properties hold. 

\begin{itemize}
	\item (Selection) For every $\{v_i,v_j\}\in E(G)$, we pick a unique edge in $\widehat H$ with one endpoint in (Row $i$, Column $i$) and the other endpoint in (Row $i$, Column $j$) and we pick a unique edge with one endpoint in (Row $j$, Column $j$) and the other endpoint in (Row $j$, Column $i$).  
	\item (Consistency 1) All the edges we pick from Row $i$ of $\widehat H$ share a common endpoint at the  position  (Row $i$, Column $i$).
	\item (Consistency 2) For any edge $e=\{a,b\}\in E(H)$ such that $c_H(a)=i$, $c_H(b)=j$, if the copy of $e$ in Row $i$ is selected in our solution then our solution contains  the copy of $e$ in Row $j$ as well.  
	
	\end{itemize} 

It is straightforward to see that a set of edges of $\widehat H$ which satisfy the stated properties imply a solution to our \PSI\ instance in an obvious way. In order to obtain our \IP\ instance, we create a variable for every edge in $\widehat H$ (or 2 for every edge in $E(H)$) and encode the properties stated above in the form of constraints. 
 We now formally define the \IP\ instance output by our reduction.

\begin{figure}[t]
\begin{center}
  \includegraphics[height=160 pt, width=340 pt]{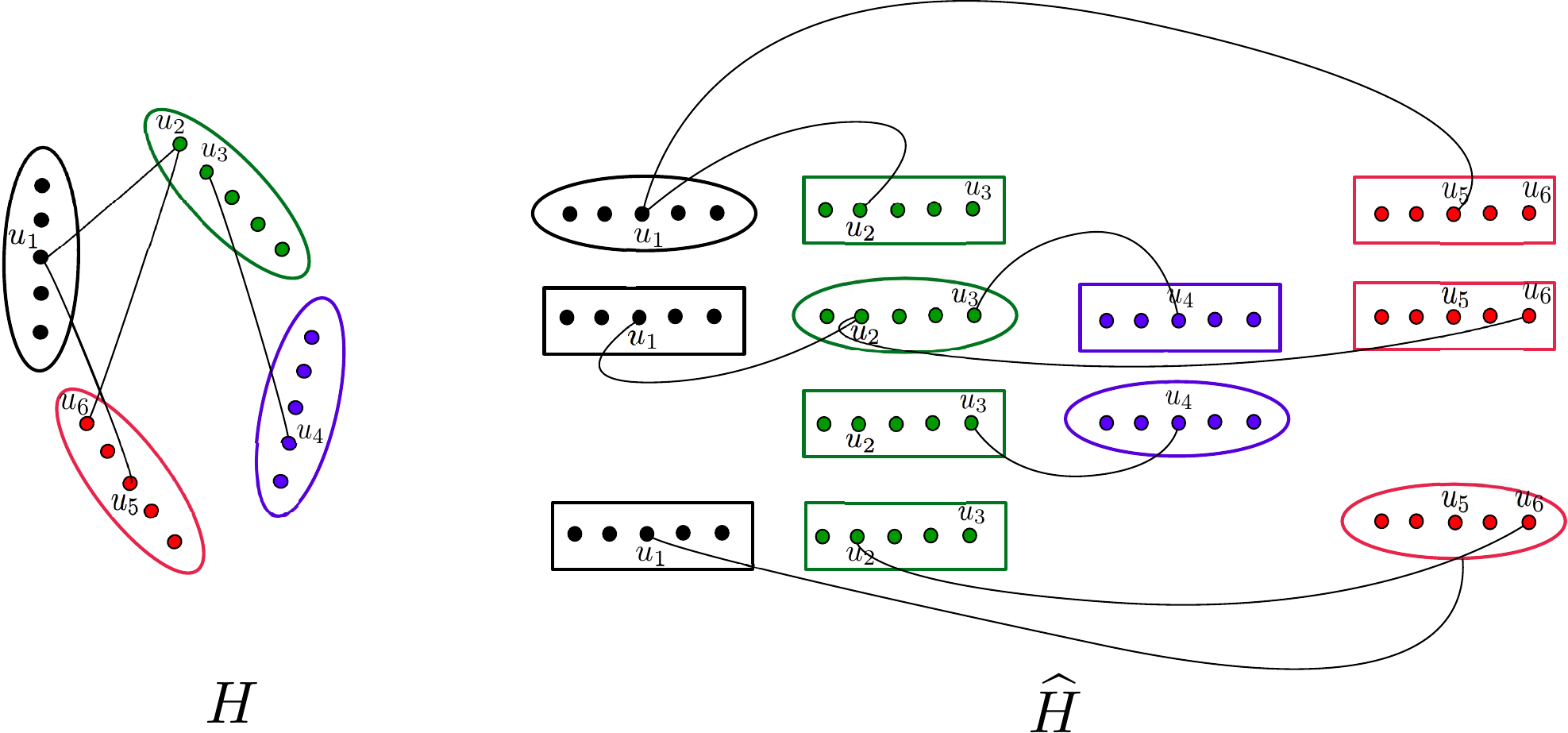}

  \caption{An illustration of the auxiliary graph $\widehat{H}$ capturing the representation of the  vertices and {\em some} edges of $H$. }
\label{fig:PSI_example_aux_graph}
  \end{center}
\end{figure}

The set of indeterminants  $x$ of the \IP instance is 
$$\left\{x(\{a,b\},c_H(a),c_H(b))\colon \{a,b\}\in E(H)  \right\}.$$ 

Notice that for any $\{a,b\}\in E(H)$, there exist an associated pair of indeterminants, namely 
$x(\{a,b\},c_H(a),c_H(b))$ and $x(\{a,b\},c_H(b),c_H(a))$.
Thus the cardinality of $x$ is upper bounded by $2 \vert E(H)\vert =2m$. 
Recall that  $\{v_1,\ldots,v_{\ell}\}=V(G)$ and $c_G(v_i)=i$ for all $i\in [\ell]$, where $\ell=\vert V(G)\vert$. 
For each $v_i\in V(G)$ we define $2d_G(v_i)-1$ many constraints as explained below. 
Let $r=d_G(v_i)$ and $N_G(v_i)=\{v_{j_1},\ldots,v_{j_r}\}$. 
The constraints for $v_i\in V(G)$ are the following.  
 For all  $q\in [r]$,
\begin{equation}
 \sum_{\substack{e\in E_H(i,j_q)}} x(e,i,j_q) = 1 \label{eqn:psi1}
 \end{equation}

 The constraints of the form above enforce the (Selection) property described in our informal summary. 
 
 For all   $q\in [r-1]$, 
\begin{equation} 
  \sum_{\substack{\{a,b\}\in E_H(i,j_q)\\ a\in c_H^{-1}(i)}} a\cdot x(\{a,b\},i,j_q) + 
   \sum_{\substack{\{a,b'\}\in E_H(i,j_{q+1})\\ a \in c_H^{-1}(i)}} (n-a)\cdot x(\{a,b'\},i,j_{q+1})
   = n \label{eqn:psi2}     
\end{equation} 
The constraints of the form above together enforce the (Consistency 1) property described in our informal summary.

For each $\{v_i,v_j\}\in E(G)$ with $i<j$, we define the following constraint in the \IP instance.  
\begin{eqnarray}
  \sum_{\substack{ \{a,b\}\in E_H(i,j)\\ a\in c_H^{-1}(i)}} \alpha(\{a,b\}) \cdot x(\{a,b\},i,j) +\sum_{ \substack{ \{a,b\}\in E_H(i,j)\\b\in c_H^{-1}(j)}} (m-\alpha(\{a,b\})) \cdot x(\{a,b\},j,i)=m  \label{eqn:psi3}
\end{eqnarray} 
The constraints of the form above together enforce the (Consistency 2) property described in our informal summary.

This completes the construction of the \IP instance $Ax=b,x\geq 0$. Notice that 
the construction of instance $Ax=b,x\geq 0$ can be done in polynomial time. Clearly, the number of rows in $A$ is 
$\vert E(G)\vert + \sum_{v\in V(G)} 2d_{G}(v)-1\leq 5k$ and number of columns in $A$ is $2m$. 
Now we prove the correctness of the reduction. 
\begin{lemma}
\label{lem:psicor}
$(G,H,c_G,c_H)$ is a {\sc Yes} instance of \PSI\ if and only if 
$Ax=b,x\geq 0$ is feasible.  Moreover, if $Ax=b,x\geq 0$ is feasible, then for any solution $x^*$, each entry  
of $x^*$ belongs to $\{0,1\}$. 
\end{lemma}
\begin{proof}
Suppose $(G,H,c_G,c_H)$ is a {\sc Yes} instance of \PSI. Let $\phi\colon V(G)\mapsto V(H)$ be a solution to $(G,H,c_G,c_H)$. Now we define a 
solution $x^*\in \{0,1\}^{2m}$ to the instance $Ax=b,x\geq 0$ of \IP. We know that for each edge $\{v_i,v_j\}\in E(G)$, $\{\phi(v_i),\phi(v_j)\}\in E(H)$. For 
each edge $\{v_i,v_j\}\in E(G)$, we set $x^*(\{\phi(v_i),\phi(v_j)\},i,j)=x^*(\{\phi(v_i),\phi(v_j)\},j,i)=1$. For every  other indeterminant, we set its value 
to $0$. Now we prove that  $Ax^*=b$.

Towards that first consider  $(\ref{eqn:psi1})$. Fix a vertex $v_i\in V(G)$ and $v_{j_q}\in N_G(v_i)$. Since $\{v_i,v_{j_q}\}\in E(G)$, $x^*(\{\phi(v_i),\phi(v_{j_q})\},i,j_q)=1$. Moreover, since $G$ is a simple graph, for any edge 
$e\in E_H(i,j_q)\setminus \{\{\phi(v_i),\phi(v_{j_q})\}\}$, $x^*(e,i,j_q)=0$. This implies that $(\ref{eqn:psi1})$ is satisfied by $x^*$. Next we consider (\ref{eqn:psi2}). Fix a vertex $v_i\in V(G)$. Let $N_G(v_i)=\{v_{j_1},\ldots,v_{j_r}\}$.
Also, fix $q\in [r-1]$. We know that $\{v_i,v_{j_{q}}\},\{v_i,v_{j_{q+1}}\}\in E(G)$. By the definition of $x^*$, we have that 
$x^*(e,i,j_q)=1$ if and only if $e=\{\phi(v_i),\phi(v_{j_q})\}$ and $x^*(e',i,{j_{q+1}})=1$ if and only if $e'=\{\phi(v_i),\phi(v_{j_{q+1}})\}$. 
 Thus we have that
\begin{eqnarray*}
  \sum_{\substack{\{a,b\}\in E_H(i,j_q)\\ a\in c_H^{-1}(i)}} a\cdot x(\{a,b\},i,j_q) + 
   \sum_{\substack{\{a,b'\}\in E_H(i,j_{q+1})\\ a \in c_H^{-1}(i)}} (n-a)\cdot x(\{a,b'\},i,j_{q+1})\\
=\phi(v_i)+(n-\phi(v_i))=n
\end{eqnarray*}
That is, $x^*$ satisfies (\ref{eqn:psi2}).  Now we consider (\ref{eqn:psi3}). Fix an edge $\{v_i,v_j\}\in E(G)$ where $i<j$. 
Again by the definition of $x^*$, we have that $x^*(e,i,j)=1$ if and only if $e=\{\phi(v_i),\phi(v_{j})\}$ and $x^*(e,j,i)=1$ if and only if $e=\{\phi(v_i),\phi(v_{j})\}$. This implies that (\ref{eqn:psi3}) is satisfied by $x^*$. Therefore $Ax=b,x\geq 0$ is feasible. 

Now we prove the converse direction of the lemma. Suppose that $Ax=b,x \geq 0$ is feasible and let $x'\in {\mathbb N}_0^{2m}$ be  a solution.
\begin{claim}
\label{claimpsi1}
Let $i,j\in [\ell]$ such that $i\neq j$ and $\{v_i,v_j\}\in E(G)$. 
Then there exists exactly one edge $e\in E_H(i,j)$ such that $x'(e,i,j)=x'(e,j,i)=1$. Moreover, for any 
$e'\in E_H(i,j)\setminus \{e\}$, $x'(e',i,j)=x'(e',j,i)=0$. 
\end{claim} 
\begin{proof}
By  (\ref{eqn:psi1}), we have that there exists exactly one edge $e_1\in E_H(i,j)$ such that 
$x'(e_1,i,j)=1$ and for all other edges $h\in E_H(i,j)\setminus \{e_1\}$, $x'(h,i,j)=0$. Again by  
(\ref{eqn:psi1}), we have that there exists exactly one edge $e_2\in E_H(i,j)$ such that 
$x'(e_2,j,i)=1$ and for all other edges $h\in E_H(i,j)\setminus \{e_2\}$, $x'(h,j,i)=0$. 
By   (\ref{eqn:psi3}), we have that $e_1=e_2$. This completes the proof of the claim. 
\end{proof}

Now we define an injection $\phi\colon V(G)\mapsto V(H)$ and prove that indeed $\phi$ is a subgraph isomorphism from $G$ to $H$. 
For any $i,j\in [\ell]$ with  $i\neq j$ and $\{v_i,v_j\}\in E(G)$ 
consider the edge $e=\{a,b\}\in E_H(i,j)$ such that $x'(\{a,b\},i,j)=x'(\{a,b\},j,i)=1$ (by Claim~\ref{claimpsi1}, there exits exactly one such edge in $E_H(i,j)$).
Let $a\in c_H^{-1}(i)$ and $b\in c_H^{-1}(j)$. Now we set $\phi(v_i)=a$ and $\phi(v_j)=b$.  
We claim that $\phi$ is well defined. Fix a vertex $v_i\in V(G)$. Let $r=d_G(v_i)$ and $N_G(v_i)=\{v_{j_1},\ldots,v_{j_r}\}$. By Claim~\ref{claimpsi1}, 
we know that for any $q\in [r]$, there exists exactly one edge $\{a_q,b_q\}\in E_H(i,j)$ such that 
$x'(\{a_q,b_q\},i,j_q)=x'(\{a_q,b_q\},j_q,i)=1$. Here, $a_q\in c_H^{-1}(i)$ and $b_q\in c_H^{-1}(j_q)$.  To prove that $\phi$ is well defined, it is enough to prove that $a_1=a_2=\ldots=a_r=\phi(v_i)$. By (\ref{eqn:psi2}), we have that for any $q\in [r-1]$, $a_q=a_{q+1}$.  
Also since $x'(\{a_q,b_q\},i,j_q)=x'(\{a_q,b_q\},j_q,i)=1$ for all $q\in [r]$, we have that $a_1=a_2=\ldots=a_r=\phi(v_i)$.  From the construction 
of $\phi$, we have that for any $i,j\in [\ell]$, $i\neq j$, $\phi(v_i)\in c_H^{-1}(i)$ and $\phi(v_j)\in c_H^{-1}(j)$. Moreover, $c_H^{-1}(i)\cap c_H^{-1}(j)=\emptyset$. This implies that $\phi$ is an injective map. 

Now we prove that $\phi$ is an isomorphism from $G$ to $H$. Since $\phi(v_i)\in c_H^{-1}(i)$ for all $i\in[\ell]$, to prove that 
$\phi$ is an isomorphism, it is enough to prove that for any edge $\{v_i,v_j\}\in V(G)$, $\{\phi(v_i),\phi(v_j)\}\in E(H)$.  
Fix an edge $\{v_i,v_j\}\in V(G)$ with $i<j$. 
By Claim~\ref{claimpsi1},  there exists exactly one edge $\{a,b\}\in E_H(i,j)$ such that 
$x'(\{a,b\},i,j)=x'(\{a,b\},j,i)=1$, where $a\in c_H^{-1}(i)$ and $b\in c_H^{-1}(j)$. From the definition of 
$\phi$, we have that $\phi(v_i)=a$ and $\phi(v_j)=b$. That is, $\{\phi(v_i),\phi(v_j)\}=\{a,b\}\in E(H)$. 

By Claim~\ref{claimpsi1}, we conclude that if $Ax=b,x\geq 0$ is feasible, then for any solution $x^*$, each entry  
of $x^*$ belongs to $\{0,1\}$. 
This completes 
the proof of the lemma. 
\end{proof}
\begin{proof}[Proof of Theorem~\ref{thm:ETHIP2}]
Let $(G,H,c_G,c_H)$ be an instance of \PSI. Let $Ax=b,x\geq 0$ be the instance of \IP\ constructed from $(G,H,c_G,c_H)$ as mentioned above. 
We know that the construction of $Ax=b,x\geq 0$ takes time polynomial in $n$, where $n=\vert V(H)\vert$. Also, we know 
that the number of rows and columns in $A$ is $\leq 5\vert E(G)\vert$ and $2\vert E(H)\vert$, respectively. Moreover, the 
maximum entry in $b$ is $\max \{\vert V(H)\vert, \vert E(H)\vert\}$. 

Suppose there is an algorithm ${\cal A}$ for \IP, running in time $f(m')(n'\cdot d')^{o\left(\frac{m'}{\log m'}\right)}$ on instances where the constraint matrix is non-negative and is of dimension $m'\times n'$, and the maximum entry in the target vector is $d'$. Then, by running ${\cal A}$ on $Ax=b,x\geq 0$ and applying Lemma~\ref{lem:psicor}, we solve \PSI\ in time $f(G)n^{o\left(\frac{k}{\log k}\right)}$. 
Thus by Lemma~\ref{lem:psimarx}, ETH fails. 
This completes the proof of the theorem.  
\end{proof}

\label{sec:lowbranchwidth}

\section{Path-width parameterization: SETH bounds }\label{sec:SETHLB}

In this section we prove Theorems~\ref{thm:lowbranchwidth} and~\ref{thm:lowentries}. 
\subsection{Overview of our reductions}
\label{subsec:overview}


We prove Theorems~\ref{thm:lowbranchwidth} and \ref{thm:lowentries} 
by giving reductions from {\sc CNF-SAT}. At this point, one might be tempted to start the reduction from {\sc $k$-CNF SAT} as seen in~\cite{CyganDLMNOPSW12}. However, the fact that in our case we also need to control the \pw of the reduced instance poses serious technical difficulties if one were to take this route. Therefore, we take a different route and reduce from {\sc CNF-SAT} which allows us to construct appropriate gadgets for propagation of consistency in our instance while simultaneously controlling the \pw. 
Moreover, 
%
%
%
%
the parameters in the 
reduced instances are required to obey certain strict conditions. For example, the reduction we give to prove 
Theorem~\ref{thm:lowbranchwidth} must output an instance of \IP, where the \pw of the 
column matroid $M(A)$ of the constraint matrix $A$  is a constant. Similarly, in the reduction 
used to prove Theorem~\ref{thm:lowentries}, we need to construct   an instance of \IP where  
the largest entry in the target vector is upper bounded by a constant. 
These stringent requirements on the parameters make the {SETH}-based reductions quite  challenging.
However, reductions under SETH can take super polynomial time---they can even take  $2^{(1-\epsilon)n}$ time for some $\epsilon>0$, where $n$ 
is the number of variables in the instance of {\sc CNF-SAT}.    This freedom to avail exponential time 
in SETH-based reductions  is used crucially in the proofs of Theorems~\ref{thm:lowbranchwidth} and \ref{thm:lowentries}.

%
%

Now we give an overview of the reduction used 
to prove  Theorem~\ref{thm:lowbranchwidth}. 
Let $\psi$ be an instance of {\sc CNF-SAT} 
with $n$ variables and $m$ clauses.  
Given $\psi$ and a fixed constant $c\geq 2$, we construct an instance  $\cm x= \tv, x\geq 0$ of \IP\ 
satisfying certain properties. Since for every $c\geq 2$, we have a different $\cm$ and $\tv$, this can be 
viewed as a family of instances of \IP. In particular our main technical lemma is the following. 


\begin{restatable}{lem}{lemtechnicalintro}
\label{lemtechnicalintro}{\em
Let $\psi$ be an instance of {\sc CNF-SAT}  with $n$ variables and $m$ clauses. Let $c\geq 2$ be a fixed integer. Then, 
in time $\OO(m^2 2^{\frac{n}{c}})$, we can  construct an  instance $\cm x=\tv,x\geq 0$,  of \IP{} with the following properties.
 \begin{enumerate} 
 \setlength{\itemsep}{-2pt}
 \item[(a.)] $\psi$ is satisfiable if and only if $\cm x=\tv, x\geq 0$ is feasible. 
\item[(b.)] The matrix $\cm$ is non-negative and has dimension $\OO(m)\times \OO(m 2^{\frac{n}{c}})$.  
\item[(c.)] The \pw of the column matroid of $\cm$ is at most $c+4$. 
\item[(d.)] The largest entry in $\tv$ is at most $2^{\lceil \frac{n}{c}\rceil}-1$.  
\end{enumerate}}
\end{restatable}


Once we have Lemma~\ref{lemtechnicalintro},  the proof of Theorem~\ref{thm:lowbranchwidth}  follows from the following observation: if we have an algorithm $\cal A$
solving \IP  in time  $f(k)(\valueB+1)^{(1-\epsilon)k}(mn)^{a}$  for some 
$\epsilon, a>0$, then we can use this algorithm to refute SETH. In particular, given an instance $\psi$ of {\sc CNF-SAT}, we 
choose an appropriate $c$ depending only on $\epsilon$ and $a$,  construct an instance 
$\cm x=\tv,x\geq 0,$ of \IP,  and run $\cal A$ on it. Our careful choice of $c$ will imply a faster algorithm for {\sc CNF-SAT}, refuting SETH. 
%
More formally, we choose $c$ to be an integer  such that $(1-\epsilon)+\frac{4(1-\epsilon)}{c}+\frac{a}{c}<1$. 
Then the total running time to test whether $\psi$ is satisfiable, is the time require to construct 
$\cm x=\tv, x\geq 0$ plus the time required by   $\cal A$ to solve the constructed instance of \IP. That is, the time required to  test whether $\psi$ is satisfiable is 
$$
\OO(m^2 2^{\frac{n}{c}})+ f(c+4)2^{\frac{n}{c} (1-\epsilon)(c+4)} 2^{\frac{a\cdot n}{c}} m^{\OO(1)} 
= 2^{\left( (1-\epsilon)+\frac{4(1-\epsilon)}{c} +\frac{a}{c} \right)n} m^{\OO(1)}
= 2^{\epsilon'n} m^{\OO(1)},
$$
 where $\epsilon'<1$ is a constant depending on the choice of $c$.
It is important to {note} that the utility of the reduction described in  Lemma~\ref{lemtechnicalintro} 
is extremely sensitive to the value of the numerical parameters involved.  In particular, even when the \pw blows up slightly, say up to $\delta c$, 
or when the largest entry in $\tv$ blows up slightly, say up to $2^{\delta \frac{n}{c}}$, for some $\delta>1$, then 
the  calculation above will \emph{not} give us the desired refutation of {SETH}.  Thus, the challenging part 
of the reduction described in Lemma~\ref{lemtechnicalintro} is making it work under these strict restrictions on the relevant parameters.   

\begin{figure}
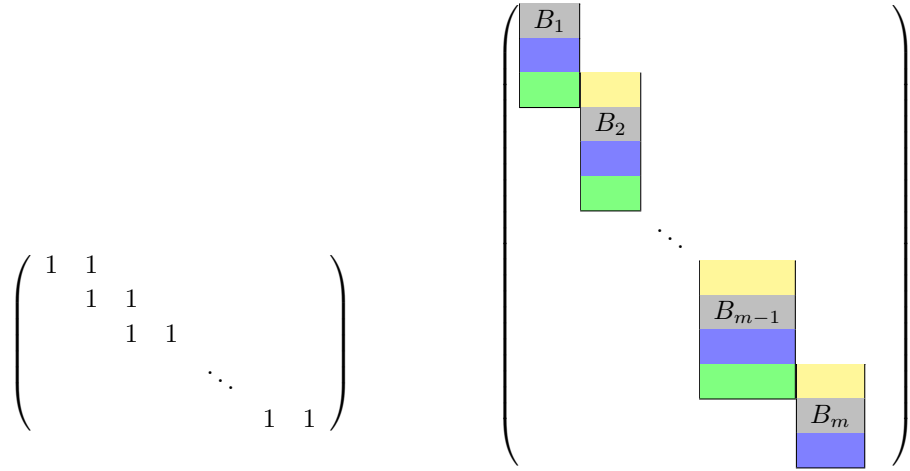

\centering
 \begin{subfigure}[b]{0.45\textwidth}
 \[
   \left(\begin{array}{ccccccc}
  1&1& \\
&1&1\\
&&1&1\\
&&&&\ddots\\
&&&&&1&1
  \end{array}\right)
 \]
 \caption{A matrix $B$ for which  \pw of its column matroid is $1$}
 \label{fig:B}
 \end{subfigure}
  \begin{subfigure}[b]{0.5\textwidth}
 \[
   \left(\begin{array}{cccccc}
 \cline{1-1}
      \multicolumn{1}{|c|}{\cellcolor{gray!50} \scalebox{1}{$B_1$}} & \\ \cline{1-1}
       \multicolumn{1}{|c|}{\cellcolor{blue!50}} & \\ \cline{1-1}
      \multicolumn{1}{|c|}{\cellcolor{green!50}} & \multicolumn{1}{|c|}{\cellcolor{yellow!50}} \\ \cline{1-1} 
     &\multicolumn{1}{|c|}{\cellcolor{gray!50} \scalebox{1}{$B_2$}} & \\ 
      & \multicolumn{1}{|c|}{\cellcolor{blue!50} } & \\ 
     & \multicolumn{1}{|c|}{\cellcolor{green!50}}   \\ \cline{2-2}
    && \ddots \\ \cline{4-4}
&&& \multicolumn{1}{|c|}{\cellcolor{yellow!50}} \\
&&&      \multicolumn{1}{|c|}{\cellcolor{gray!50} \scalebox{1}{$B_{m-1}$}}  \\ 
&&&      \multicolumn{1}{|c|}{\cellcolor{blue!50} }  \\ \cline{5-5}
 &&&     \multicolumn{1}{|c|}{\cellcolor{green!50}} & \multicolumn{1}{|c|}{\cellcolor{yellow!50}} \\ \cline{4-4} 
 & &&   &\multicolumn{1}{|c|}{\cellcolor{gray!50} \scalebox{1}{$B_m$}} & \\ 
  & &&  & \multicolumn{1}{|c|}{\cellcolor{blue!50}}   \\ \cline{5-5}

  \end{array}\right)
 %
 \]
 \caption{A pictorial representation of the matrix $\cm$.}
 \label{fig:A}
 \end{subfigure}
\caption{Comparison of $\cm$ with a low path-width matrix.}
 \label{fig:comparison}
\end{figure}


As stated in Lemma~\ref{lemtechnicalintro},  in our reduction, we need to obtain a constraint matrix with small \pw. An important first step towards this is understanding what a matrix of small \pw looks like.
We first give an intuitive description of the structure of such matrices.  
Let $A$ be a  $m\times n$ matrix of small \pw and 
let $M(A)$ be the column matroid of $A$. 
For any $i\in \{1,\ldots,n-1\}$, let $A|\{1,\ldots i\}$ denote the set of columns (or vectors) in $A$ whose index is at most $i$ (that is, the first $i$ columns)  and let ${A}|\{i+1,\ldots n\}$ denote the set of columns with index strictly 
greater than $i$. 
The \pw 
of $M(A)$ is at most 
$$\max_i \operatorname{dim}\langle  \operatorname{span}(A|\{1,\dots, i\})\cap\operatorname{span}(A|\{i+1, \dots, n\})\rangle +1 .$$
Hence, in order to obtain a bound on the pathwidth, it is sufficient to bound  $\operatorname{dim}\langle  \operatorname{span}(A|\{1,\dots, i\})\cap\operatorname{span}(A|\{i+1, \dots, n\})\rangle$ for every $i\in [n]$.
 Consider for example, the matrix $B$ given in Figure~\ref{fig:B}. The \pw of $M(B)$ is clearly at most $1$. In \emph{our} reduced instance, the constructed constraint matrix $\cm$ will be an  appropriate extension of $B$. That is $\cm$ will have the ``same form'' as $B$ but with each $1$ replaced by a submatrix of order $\OO(c)\times n'$ for some $n'$.  See Fig.~\ref{fig:A} for a pictorial  representation of $\cm$.


The  construction  used in Lemma~\ref{lemtechnicalintro}  takes as input an instance $\psi$  of {\sc CNF-SAT} with $n$ variables and  a fixed integer $c\geq 2$,  and outputs an  instance $\cm x=\tv,x\geq 0$,  of \IP{}, that satisfies all four properties of the lemma. Let $X$ denote the set of variables in the input {\sc CNF}-formula $\psi=C_1\wedge C_2\wedge\ldots \wedge C_{m}$.  For the purposes of the present 
discussion we assume that $c$ divides $n$. We partition the variable set $X$ into 
$c$ blocks $X_0,\ldots,X_{c-1}$, each of size $\frac{n}{c}$.
Let ${\cal X}_i$, $i \in \{0,\ldots, c-1\}$, denote the set of assignments of variables corresponding to 
$X_i$. Set $\ell=\frac{n}{c}$ and $L=2^{\ell}$.  Clearly, the size of ${\cal X}_i$ is upper bounded by $2^{\frac{n}{c}}=2^\ell= L$.   
 We denote the assignments in ${\cal X}_i$  by $\phi_{0}(X_{\ii}),\phi_1(X_{\ii}),\ldots, \phi_{L-1}(X_{\ii})$. 
To construct  the matrix $\cm$, we view ``each of these assignments as a different assignment 
for each clause''.  In other words we have separate sets of variables
in the constraints 
corresponding to different pairs $(C_r,X_\ii)$, where $C_r$ is a clause and  $X_\ii$ is a block in the partition of $X$. 
That is for each clause $C_r$ and block $X_\ii$, we have variables $\{y_{C_r,\ii,a} ~a\in \ZZ{2L}~\}$. In other words 
for each $C_r$ and assignment $\phi_a(X_i)$, $a\in \ZZ{L}$, we have two variables $y_{C_r,\ii,2a}$ and $y_{C_r,\ii,2a+1}$. 
For any clause $C_r$, $\ii\in \ZZ{c}$ and $a\in \ZZ{2L}$, assigning value $1$ to $y_{C,\ii,a}$ corresponds to choosing an assignment 
$\phi_{\lfloor \frac{a}{2}\rfloor}(X_\ii)$ for $X_\ii$. 
In our reduction we will create the following set of constraints.   

\begin{eqnarray}
\sum_{\substack{i\in [c],a \in \ZZ{2L} \mbox{ such that} \\ 
a \mbox{ is even and} \\ 
\phi_{\lfloor \frac{a}{2}\rfloor}(X_{\ii}) \mbox{ satisfies } C}} y_{C,i,a}&=&1 \qquad  \mbox{ for all } C\in \CC\label{eqn:sat0}\label{eqn:sat}\\
\sum_{\substack{a \in\ZZ{2L} }} y_{C,i,a}&=&1 \qquad  \mbox{ for all } C\in \CC \mbox{ and } i\in\ZZ{c}\label{eqn:onlyone0}\label{eqn:onlyone}
\end{eqnarray}
Equation~(\ref{eqn:sat0}) takes care of satisfiability of clauses, while Equation~(\ref{eqn:onlyone0}) allows us to pick only one 
assignment from  $\{\phi_{0}(X_{\ii}),\phi_1(X_{\ii}),\ldots, \phi_{L-1}(X_{\ii})\}$ per clause $C$ and block $X_{\ii}$. 
Note that this implies that we will choose an assignment in ${\cal X}_\ii$ for each clause $C_r$. 
That way  we might choose $m$ assignments from  ${\cal X}_\ii$ corresponding to $m$ different clauses. However, for the backward direction of the proof, it is important that we choose the {\em same} assignment from ${\cal X}_\ii$ for each clause. This will ensure that we have selected an assignment to the variables in $X_\ii$. Towards this  we will have a third set of constraints as follows. 
\begin{eqnarray}
\sum_{\substack{a \in \ZZ{2L} }} \left(\lfloor \frac{a}{2}\rfloor\cdot y_{C_r,i,a}\right)   + \left((L-1-\lfloor \frac{a}{2}\rfloor) y_{C_{r+1},i,a}\right)=L-1 \;  \mbox{ for all } r\in [m-1] \mbox{ and } i\in\ZZ{c}\;\label{eqn:consistancy0}\label{eqn:consistancy}
\end{eqnarray}


Equation~(\ref{eqn:consistancy0}) enforce  consistencies of assignments of blocks across clauses in a {\em sequential} manner. That is, for any block 
$X_i$, we make sure that the two variables set to $1$ corresponding to $(C_{r},X_i)$ and $(C_{r+1},X_i)$ are consistent for any $r\in \{1, \ldots, m-1\}$, as opposed to checking the consistency for every pair  
$(C_{r},X_i)$ and $(C_{r'},X_i)$ for $r\neq r'$. Thus in some sense these consistencies {\em propagate}.  
%
Furthermore, the idea of making consistency in a sequential manner also allows us to bound the \pw of column matroid of $\cm$ by $c+4$.

The proof technique for Theorem~\ref{thm:lowentries} is similar to that  for Theorem~\ref{thm:lowbranchwidth}. This is achieved 
by modifying  the matrix $\cm$ constructed in the reduction described for Lemma~\ref{lemtechnicalintro}. 
The largest entry in $\cm$ is $2^{\frac{n}{c}}-1$ (see Equation~(\ref{eqn:consistancy0})). So each of these values can be represented 
by a binary string of length at most $\ell=\frac{n}{c}$.  We remove each row, say row indexed by $\gamma$, with entries greater than $1$ and replace it with $\frac{n}{c}$ rows,  $\gamma_{1},\ldots,\gamma_{\ell}$. Where, for any 
$j$, if the value $\cm[\gamma,j]=W$ then $\cm[\gamma_{k},j]=\eta_k$, where $\eta_k$ is the $k^{th}$ bit in the 
$\ell$-sized binary representation of $W$.  
This modification reduces the largest entry in $\cm$ to  $1$ and increases the \pw from constant to approximately $n$. Finally, we set all the entries in $\tv$ to be $1$. This concludes the overview of our reductions and we now proceed to a  detailed exposition.

\subsection{Proof of Theorem~\ref{thm:lowbranchwidth}}\label{subsec:lowbr}

In this section we provide a Proof of Theorem~\ref{thm:lowbranchwidth}, which states that unless SETH fails, \IP{} with non-negative matrix $A$ cannot be solved in time $f(k)(\valueB+1)^{(1-\epsilon)k}(mn)^{\OO(1)}$ for any function 
$f$ and $\epsilon>0$,   
where $d=\max\{b[1],\ldots,b[m]\}$ and $k$ is the \pw of the column matroid of $A$. 


Towards the proof of Theorem~\ref{thm:lowbranchwidth}, we first present the proof of our main technical lemma 
(Lemma~\ref{lemtechnicalintro}), which we restate here for the sake of completeness.

\lemtechnicalintro*

Let $\psi=C_1\wedge C_2\wedge\ldots \wedge C_m$ be 
an instance  of {\sc CNF-SAT} with variable set $X=\{x_1,x_2,\ldots,x_n\}$ and let $c\geq 2$ be a fixed constant given in the statement of Lemma~\ref{lemtechnicalintro}. 
We construct the instance $\cm x=\tv,x\geq 0$ of \IP as follows.

\medskip
\noindent
{\bf Construction.} 
Let $\CC=\{C_1,\ldots,C_m\}$. 
Without loss of generality, we assume that $n$ is divisible 
by $c$, otherwise we add at most $c$  dummy variables to $X$ such that $\vert X\vert$  
is divisible by $c$. We divide $X$ into $c$ blocks $X_0,X_1,\ldots,X_{c-1}$. That is 
$X_{\ii}=\{x_{\frac{\ii \cdot n}{c}+1},x_{\frac{\ii\cdot n}{c}+2},\ldots, x_{\frac{(\ii+1)\cdot n}{c}}\}$ for each $\ii\in \ZZ{c}$. 
Let $\ell=\frac{n}{c}$ and $L=2^{\ell}$. 
For each block $X_{\ii}$, there are exactly $2^{\ell}$ assignments. We denote these 
assignments by $\phi_{0}(X_{\ii}),\phi_1(X_{\ii}),\ldots, \phi_{L-1}(X_{\ii})$. 

Now, we create $m\cdot c\cdot 2^{\ell+1}$ variables; they are named $y_{C,i,a}$, where 
$C\in \CC$, $i\in \ZZ{c}$ and $a\in \ZZ{2L}=\ZZ{2^{\ell+1}}$. In other words, for a clause $C$, a block $X_{\ii}$ and an assignment 
$\phi_a(X_{\ii})$, we create two variables; they are $y_{C,\ii,2a}$ and $y_{C,\ii,2a+1}$.  
Then, we create the \IP constraints given by Equations~(\ref{eqn:sat}), (\ref{eqn:onlyone}), and (\ref{eqn:consistancy}). 

This completes the construction of \IP instance. 
Let $\cm y=\tv$ be the  \IP instance defined using Equations~(\ref{eqn:sat}), (\ref{eqn:onlyone}), and (\ref{eqn:consistancy}).   
The purpose of Equation~(\ref{eqn:sat}) is to ensure satisfiability of all the clauses. 
Because of Equation~(\ref{eqn:onlyone}), for each clause $C$ and for each block $X_i$, we select only one assignment. Notice, 
that, so far it is allowed to choose many assignments from a block $X_i$, for different clauses. To ensure the consistency 
of assignments in each block across clauses, we added a system of constraints (Equation~(\ref{eqn:consistancy})). Equation~(\ref{eqn:consistancy}) 
ensures the consistency of assignments in the adjacent clauses (in the order $C_1,\ldots,C_m$).  Thus, the consistency of assignments 
propagates in a  sequential manner. Notice that number constraints defined by  Equations~(\ref{eqn:sat}), (\ref{eqn:onlyone}), and (\ref{eqn:consistancy}) are $m$, $m\cdot c$ and $(m-1)\cdot c$, respectively. The number of variables is $m\cdot c \cdot 2^{\ell+1}$. 
Also notice that all the coefficients in  Equations~(\ref{eqn:sat}), (\ref{eqn:onlyone}) and (\ref{eqn:consistancy}) are non-negative. 
This implies that $\cm$ is non-negative and has dimension $\OO(m)\times \OO(m 2^{\frac{n}{c}})$. Thus, the property $(b.)$ of Lemma~\ref{lemtechnicalintro} is satisfied. The largest entry in $\tv$ is $L-1= 2^{\lceil \frac{n}{c}\rceil}-1$ (see Equation~(\ref{eqn:consistancy})) and hence the property $(d.)$ of Lemma~\ref{lemtechnicalintro} is satisfied. Now we prove property $(a.)$ of  Lemma~\ref{lemtechnicalintro}.  

\begin{center}
\textbf{
 we simplify the notation by using $A$ instead of $\cm $ and $b$ instead of  $ \tv$.
 }
\end{center}

 \begin{lemma}
\label{lemma:brachwidthcorrect}
 Formula  $\psi$ is satisfiable if and only if there exists $y^*\in \Z^{n'}$ such that $A y^*=b$.
 where $n'= m\cdot c \cdot 2^{\ell+1}$, the number of columns in $A$. 
 \end{lemma}

\begin{proof}
Let $Y=\{y_{C,i,a}~|~C\in \CC,i\in \ZZ{c},a\in \ZZ{2L}\}$. 
Suppose $\psi$ is satisfiable.  
We need to show that there is an assignment of non-negative integer values to the variables 
in $Y$ such that 
Equations~(\ref{eqn:sat}), (\ref{eqn:onlyone}) and (\ref{eqn:consistancy}) are 
satisfied. 
Let $\phi$ be a satisfying assignment of $\psi$. Then, 
there exist $a_0,a_1,\ldots,a_{c-1}\in \ZZ{L}$ such that 
 $\phi$ is the union of $\phi_{a_0}(X_0),\phi_{a_1}(X_1),\ldots,\phi_{a_{c-1}}(X_{c-1})$. 
Any clause $C\in \CC$  is satisfied by at least one of the assignments 
 $\phi_{a_0}(X_0),\phi_{a_1}(X_1),\ldots,\phi_{a_{c-1}}(X_{c-1})$. For each 
 $C$, we fix an arbitrary $\ii \in \ZZ{c}$ such that 
 the assignment $\phi_{a_{\ii}}(X_{\ii})$ satisfies clause $C$. 
 Let $\alpha$ be a function which fixes these assignments for each clause. 
 That is, $\alpha: \CC \rightarrow \ZZ{c}$ such that the assignment $\phi_{a_{\alpha(C)}}(X_{\alpha(C)})$ 
 satisfies the clause $C$ for every $C\in\CC$. 
 Now we assign  values to $Y$ and 
 prove that these assignment satisfy  
Equations~(\ref{eqn:sat}), (\ref{eqn:onlyone}) and (\ref{eqn:consistancy}). 

\begin{equation}
y_{C,i,a} =
\left\{ \begin{array}{ll}
1, & \mbox{if }  \alpha(C)=i \mbox{ and } a \mbox{ is even and } \lfloor \frac{a}{2}\rfloor=a_i\\ 
1, & \mbox{if }  \alpha(C)\neq i \mbox{ and } a \mbox{ is odd and } \lfloor \frac{a}{2}\rfloor=a_i\\
0, & \mbox{otherwise.}
\end{array}\right.\label{eqn:assignment}
\end{equation}
Notice that, by Equation~(\ref{eqn:assignment}), for any fixed $C\in \CC$, exactly $c$ variables from 
$\{y_{C,i,a}~|~i\in \ZZ{c},a\in [2^{\ell+1}]\}$ is set to $1$. 
They are $y_{C,\alpha(C),2a_{\alpha(C)}}$ and the variables in the set $Y_C=\{y_{C,i,2a_i+1}~|~i\neq \alpha(C)\}$. 
This implies that in Equation~(\ref{eqn:sat}), only variable is set to $1$, and hence Equation~(\ref{eqn:sat}) is 
satisfied. 
Now consider Equation~(\ref{eqn:onlyone}) for any fixed $C\in \CC$ and $i\in \ZZ{c}$. By equation~(\ref{eqn:assignment}), 
exactly one variable from $\{y_{C,i,a}~|~a\in \ZZ{2L}\}$ is set to $1$, and hence Equation~(\ref{eqn:onlyone}) is satisfied. 
Now consider Equation~(\ref{eqn:consistancy}) for fixed $r\in [m-1]$ and $i\in \ZZ{c}$. By Equation~(\ref{eqn:assignment}), 
exactly one variable from each set $\{y_{C_r,i,a}~|~a\in \ZZ{2L}\}$ and $\{y_{C_{r+1},i,a}~|~a\in \ZZ{2L}\}$ 
are set to $1$; they are one variable each from $\{y_{C_r,i, 2a_i},y_{C_r,i, 2a_i+1}\}$ and $\{y_{C_{r+1},i, 2a_i},y_{C_{r+1},i, 2a_i+1}\}$. 
So we get the following when we substitute values for $Y$ in  Equation~(\ref{eqn:consistancy}). 
\begin{eqnarray*}
\sum_{\substack{a \in \ZZ{2L} }} \left(\lfloor \frac{a}{2}\rfloor\cdot y_{C_r,i,a}\right)   + \left((L-1-\lfloor \frac{a}{2}\rfloor)\cdot y_{C_{r+1},i,a}\right) = a_i + L-1 - a_i = L-1 
\end{eqnarray*}
Hence, Equation~(\ref{eqn:consistancy}) is satisfied by the assignments given in Equation~(\ref{eqn:assignment}). 

Now we need to prove the converse direction.  Suppose there are non-negative integer assignments to $Y$ such that 
Equations~(\ref{eqn:sat}), (\ref{eqn:onlyone}) and (\ref{eqn:consistancy}) are satisfied. Now we need to show that 
$\psi$ is satisfiable. Because of Equation~(\ref{eqn:onlyone}) all the variables in $Y$ are set to $0$ or $1$. 
We will extract a satisfying assignment from the  values assigned to variables in $Y$. Towards that, first we prove the following claim. 
\begin{claim}\label{claim:extractsat}
Let $y_{C_1,i,a}=1$ for some $i\in \ZZ{c}$ 
and $a\in \ZZ{2L}$. 
Then, for any  $C'\in \CC$,  exactly one among $\{y_{C',i,2\lfloor \frac{a}{2}\rfloor}, y_{C',i,2\lfloor \frac{a}{2}\rfloor+1}\}$ is 
set to $1$. 
\end{claim} 
\begin{proof} 
Towards the proof, we first show that if $y_{C_r,i,a}=1$ for some $r\in [m-1]$, 
then  exactly one among $\{y_{C_{r+1},i,2\lfloor \frac{a}{2}\rfloor}, y_{C_{r+1},i,2\lfloor \frac{a}{2}\rfloor+1}\}$ is set to $1$. 
%
 By Equation~(\ref{eqn:onlyone}) and the fact that  $y_{C_r,i,a}=1$, 
we get that 
\begin{equation}
\sum_{\substack{a' \in \ZZ{2L} }} \left(\lfloor \frac{a'}{2}\rfloor\cdot y_{C_r,i,a'}\right)=\lfloor \frac{a}{2}\rfloor.\label{eqn:aby2}
\end{equation} 
Equations~(\ref{eqn:consistancy}) and (\ref{eqn:aby2}) implies that   

\begin{equation}
\sum_{\substack{a' \in \ZZ{2L} }}  \left((L-1-\lfloor \frac{a'}{2}\rfloor)\cdot y_{C_{r+1},i,a'}\right)= L-1- \lfloor \frac{a}{2}\rfloor. 
\label{eqn:aby2second}
\end{equation}

By Equations~(\ref{eqn:onlyone}) and (\ref{eqn:aby2second}), 
we get that exactly one among $\{y_{C_{r+1},i,2\lfloor \frac{a}{2}\rfloor}, y_{C_{r+1},i,2\lfloor \frac{a}{2}\rfloor+1}\}$ is 
set to $1$. 
Thus, by applying the above arguments for $i=1,2,\ldots,m-1$, we get that for any 
 $C'\in \CC\setminus \{C_1\}$,  exactly one among $\{y_{C',i,2\lfloor \frac{a}{2}\rfloor}, y_{C',i,2\lfloor \frac{a}{2}\rfloor+1}\}$ is 
set to $1$. 

Suppose $C'=C_1$. Then, by Equation~(\ref{eqn:onlyone}) and the assumption that $y_{C_1,i,a}=1$, exactly one among $\{y_{C_1,i,2\lfloor \frac{a}{2}\rfloor}, y_{C_1,i,2\lfloor \frac{a}{2}\rfloor+1}\}$ is set to $1$.
\end{proof}
Now we define a satisfying assignment for $\psi$. Towards that we give assignments for each blocks $X_0,\ldots,X_{c-1}$, such that 
the union of these assignments satisfies $\psi$. Fix any block $X_i$. By Equation~(\ref{eqn:onlyone}), exactly one among 
$\{y_{C_1,i,a}~|~a\in \ZZ{2L}\}$ 
is set to $1$. Let $a_i\in \ZZ{2L}$ such that $y_{C_1,i,a_i}=1$. Then we choose the assignment 
$\phi_{\lfloor \frac{a_i}{2}\rfloor}(X_i)$ for $X_i$.  Let $\phi$ be the assignment of $X$ which is the union of  $\psi_{\lfloor \frac{a_1}{2}\rfloor}(X_1)$, 
$\psi_{\lfloor \frac{a_2}{2}\rfloor}(X_2)$,\ldots,$\psi_{\lfloor \frac{a_{c-1}}{2}\rfloor}(X_{c-1})$.  By Equation~(\ref{eqn:sat}) 
and Claim~\ref{claim:extractsat}, $\phi$ satisfies all the clauses in $\CC$ and hence $\psi$ is satisfiable. 
\end{proof}

Now we need to prove property $(c.)$ of  Lemma~\ref{lemtechnicalintro}.   That is the \pw of 
$A$ is at most $c+4$. Towards that we need to understand the structure of matrix $A$. We decompose 
the matrix $A$ into $m$ disjoint submatrices $B_1,\ldots B_m$ which are disjoint and cover all the non-zero entries in the matrix $A$. 
First we define some notations and fix the column indices  of $A$ corresponding the the variables in the 
constraints. 
Let  $Y$ denote the set $\{y_{C,i,a}~|~C\in \CC,i\in \ZZ{c},a\in \ZZ{2L}\}$ of variables in the constraints defined by  
Equations~(\ref{eqn:sat}), (\ref{eqn:onlyone}) and (\ref{eqn:consistancy}). 
These variables  can be partitioned into $\biguplus_{C\in \CC}Y_C$, where 
$Y_C=\{y_{C,i,a}~|~i\in \ZZ{c},a\in \ZZ{2L}\}$. 
Further for each $C\in \CC$, $Y_C$ can be partitioned into $\bigcup_{i\in{\ZZ{c}}}Y_{C,i}$, 
where $Y_{C,i}=\{y_{C,i,a}~|~a\in \ZZ{2L}\}$. 
The set of columns indexed by $[r\cdot c\dot 2^{\ell+1}]\setminus [(r-1)\cdot c \cdot 2^{\ell+1}]$, for any 
$r\in[m]$, corresponds to the set of variables in $Y_{C_r}$. Among the set of columns corresponding to 
$Y_C$, the first $2^{\ell+1}$ columns corresponds to the variables in $Y_{C,1}$, second 
$2^{\ell+1}$ columns corresponds to the variables in $Y_{C,2}$, and so on. Among the set of columns 
corresponds to $Y_{C,i}$ for any $C\in \CC$ and $i\in \ZZ{c}$, the first two columns corresponds to the variable  
$y_{C,i,0}$ and $y_{C,i,1}$, and second two columns corresponds to the variables 
$y_{C,i,2}$ and $y_{C,i,3}$, and so on.

Now we move to the description of $B_j$, $j\in [m]$. The matrix 
$B_j$ will cover the coefficients of  $Y_{C_j}$ in Equations~(\ref{eqn:sat}), (\ref{eqn:onlyone}) and (\ref{eqn:consistancy}). 
In other words $B_j$ covers the non-zero entries in the columns corresponding to $Y_{C_j}$, i.e, in the columns of 
$A$ indexed by $[j \cdot c \cdot 2^{\ell+1}]\setminus [(j-1) \cdot c \cdot 2^{\ell+1}]$.   
Now we explain these submatrices. 
Each matrix $B_j$ has $c\cdot 2^{\ell+1}$ columns; each of them corresponds to a variable in 
$Y_{C_j}$. Each row in $A$ corresponds to a constraint in the system of equations defined by Equations~(\ref{eqn:sat}), (\ref{eqn:onlyone}) and (\ref{eqn:consistancy}). So we use notations $f(C_1),\ldots f(C_m)$ to represents the constraints defined by 
Equations~(\ref{eqn:sat}). Similarly we use notations  $\{s(C,i)~|~~C\in \CC,i\in \ZZ{c}\}$ and $\{t(C,i)~|~~C\in \CC,i\in \ZZ{c}\}$ 
to represents the constraints defined by Equations~(\ref{eqn:onlyone}) and (\ref{eqn:consistancy}), respectively. 
%
%



\medskip\noindent\emph{Matrix $B_1$.}
Matrix $B_1$ is of dimension $(2c+1)\times (c\cdot 2^{\ell+1})$.   
In the first row of $B_1$, we have coefficients of $Y_{C_1}$ from 
$f(C_1)$.
For $j\in [c]$, the rows indexed by $j+1$ and $c+j+1$ are defined as follows.   
In the $(j+1)^{st}$ row of $B_1$, we have coefficients of $Y_{C_1}$ from $s(C_1,j)$ while 
in the $(c+j+1)^{st}$ row of $B_1$, we have coefficients of $Y_{C_1}$ from $t(C_1,j)$. 
That is the entries of $B_1$ are as follows, where $i\in \ZZ{c}$ and $a\in \ZZ{L}$.  
\begin{eqnarray}
&&B_1[1, {\ii}\cdot 2^{\ell+1}+2a+1]=
\left\{ \begin{array}{ll}
1 & \mbox{if } \phi_{a}(X_{\ii}) \mbox{ satisfies } C_1,\\
0 & \mbox{otherwise.}
\end{array}\right. 
 \label{eqn:B1:evaluation}\\
&&B_1[1, {\ii}\cdot 2^{\ell+1}+2a+2]= 0, \text{ and}  \label{eqn:B1:evaluation0}\\
&&B_1[2+{\ii}, {\ii}\cdot 2^{\ell+1}+2a+1]=  B_1[2+{\ii}, {\ii}\cdot 2^{\ell+1}+2a+2] = 1,\label{eqn:B1:selection}\\
&& B_1[c+2+{\ii}, {\ii}\cdot 2^{\ell+1}+2a+1]= B_1[c+2+{\ii}, {\ii}\cdot 2^{\ell+1}+2a+2] =a, \label{eqn:B1:postassignement}
\end{eqnarray}

Here, Equations~(\ref{eqn:B1:evaluation}) and (\ref{eqn:B1:evaluation0}), follows from Equation~(\ref{eqn:sat}).  
Equations~(\ref{eqn:B1:selection}) and (\ref{eqn:B1:postassignement}) follows from Equation~(\ref{eqn:onlyone}) 
and (\ref{eqn:consistancy}), respectively. 
All other entries in $B_1$ are zeros. 
That is,  for all $\ii,\ip \in \ZZ{c}$ and $g\in [2^{\ell+1}]$ such that $\ii \neq \ip$, 
\begin{eqnarray}
&& B_1[2+\ii, {\ip}\cdot 2^{\ell+1}+g]=B_1[c+2+\ii, {\ip}\cdot 2^{\ell+1}+g]  =0, \label{eqn:B1:zero}
\end{eqnarray}

This completes the definition of $B_1$. By their role in the reduction, 
the matrix $B_1$ is partitioned in to three parts.  The first row is called the {\em evaluation part}  of $B_1$. 
The part composed of rows indexed by $2,3,\ldots,c+1$ is called {\em selection part} and 
the  part composed of last $c$ rows is called {\em successor matching part} (See Figure~\ref{fig:Br:portions}).

\medskip\noindent\emph{Matrices $B_r$ for  $1<r<m$.} 
Matrix $B_r$ is of dimension $(3c+1)\times (c\cdot 2^{\ell+1})$.  
The first $c$ rows are defined by Equation~(\ref{eqn:consistancy}).
For $j\in [c]$, in $i^{th}$ row, we have coefficients of $Y_{C_r}$ from 
$t(C_{r-1},i)$.    
In the $(c+1)^{st}$  row of $B_r$, we have coefficients of $Y_{C_r}$ from 
$f(C_r)$.
For $i\in [c]$, the rows indexed by $c+1+i$ and $2c+1+i$ are defined as follows.   
In the $(c+1+i)^{th}$ row of $B_r$, we have coefficients of $Y_{C_r}$ from $s(C_r,i)$ while 
in the $(2c+1+i)^{th}$ row of $B_r$, we have coefficients of $Y_{C_r}$ from $t(C_r,i)$.  
This completes the definition of $B_r$. By their role in the reduction, 
the matrix $B_r$ is partitioned in to four parts.  
The part composed of the first $c$ rows is called the {\em predecessor matching part}.
The part composed of the row indexed by $c+1$ is called the {\em evaluation part}  of $B_1$. 
The part composed of rows indexed by $c+2,c+3,\ldots,2c+1$ is called {\em selection part} and 
the  part composed of last $c$ rows is called {\em successor matching part} (For illustration see Fig.~\ref{fig:Br:portionsA}).
That is the entries of $B_1$ are  as follows, where $i\in \ZZ{c}$ and $a\in \ZZ{L}$.

The predecessor matching part is defined by 
\begin{equation}
B_r[{\ii}+1, {\ii}\cdot 2^{\ell+1}+2a+1]=B_r[{\ii}+1, {\ii}\cdot 2^{\ell+1}+2a+2]=L-1-a.  \label{eqn:Br:preassignement}
\end{equation}

The evaluation part is  defined by 
\begin{equation}
B_r[c+1, {\ii}\cdot 2^{\ell+1}+2a+2]=0,  \label{eqn:Br:evaluation0}
\end{equation}
and 
\begin{equation}
B_r[c+1, {\ii}\cdot 2^{\ell+1}+2a+1]=
\left\{ \begin{array}{ll}
1, & \mbox{if } \phi_{a}(X_{\ii}) \mbox{ satisfies } C_r,\\
0, & \mbox{otherwise.}
\end{array}\right. \label{eqn:Br:evaluation}
\end{equation}
The selection part for $B_r$ is defined  as 
\begin{eqnarray}
&&B_r[c+2+{\ii}, {\ii}\cdot 2^{\ell+1}+2a+1]=B_r[c+2+{\ii}, {\ii}\cdot 2^{\ell+1}+2a+2]=1,  \label{eqn:Br:selection}
\end{eqnarray}
The successor matching part for $B_r$ is defined as 
\begin{eqnarray}
&&B_r[2c+2+{\ii}, {\ii}\cdot 2^{\ell+1}+2a+1]=B_r[2c+2+{\ii}, {\ii}\cdot 2^{\ell+1}+2a+2]=j,  \label{eqn:Br:postssignement}
\end{eqnarray}

All other entries in $B_r$, which are not listed above, are  zero. 
That is,  for all $\ii,\ip \in \ZZ{c}$ and $g\in [2^{\ell+1}]$ such that $\ii \neq \ip$, 
\begin{eqnarray}
&&B_r[{\ii}+1, {\ip}\cdot 2^{\ell+1}+g]=0 \label{eqn:Br:zero1}, \\ 
&&B_r[c+2+\ii, {\ip}\cdot 2^{\ell+1}+g]=0,  \text{ and } \label{eqn:Br:zeroX}\\
&&B_r[2c+2+\ii, {\ip}\cdot 2^{\ell+1}+g]=0. \label{eqn:Br:zero2}
\end{eqnarray}
For an example, see Figure~\ref{fig:Br}.


\begin{figure}[t]
\centering
\begin{subfigure}[b]{0.3\textwidth}
\(
  \left(\begin{array}{ccccc}
     \multicolumn{5}{c}{\cellcolor{gray!50}\scalebox{1}{evaluation part}}\\\hline
     \multicolumn{5}{c}{\cellcolor{blue!50}} \\
     \multicolumn{5}{c}{\cellcolor{blue!50}\scalebox{1}{selection part}}\\
     \multicolumn{5}{c}{\cellcolor{blue!50}} \\\hline
     \multicolumn{5}{c}{\cellcolor{green!50}} \\
     \multicolumn{5}{c}{\cellcolor{green!50}\scalebox{1}{successor matching part}}\\
     \multicolumn{5}{c}{\cellcolor{green!50}} \\\hline
  \end{array}\right)
\)
\caption{Parts of $B_1$.}
\label{fig:Br:portions}
\end{subfigure}
\hspace{0.4cm}
\begin{subfigure}[b]{0.3\textwidth}
\(
  \left(\begin{array}{ccccc}
     \multicolumn{5}{c}{\cellcolor{yellow!50}} \\
     \multicolumn{5}{c}{\cellcolor{yellow!50}\scalebox{1}{predecessor matching part}}\\
     \multicolumn{5}{c}{\cellcolor{yellow!50}} \\\hline
     \multicolumn{5}{c}{\cellcolor{gray!50}\scalebox{1}{evaluation part}}\\\hline
     \multicolumn{5}{c}{\cellcolor{blue!50}} \\
     \multicolumn{5}{c}{\cellcolor{blue!50}\scalebox{1}{selection part}}\\
     \multicolumn{5}{c}{\cellcolor{blue!50}} \\\hline
     \multicolumn{5}{c}{\cellcolor{green!50}} \\
     \multicolumn{5}{c}{\cellcolor{green!50}\scalebox{1}{successor matching part}}\\
     \multicolumn{5}{c}{\cellcolor{green!50}} \\\hline
  \end{array}\right)
\)
\caption{Parts  of  $B_r$ for $1<r\leq m$.}
\label{fig:Br:portionsA}
\end{subfigure}
\hspace{0.5cm}
\begin{subfigure}[b]{0.3\textwidth}
\(
  \left(\begin{array}{cccccc}
     \multicolumn{6}{c}{\cellcolor{gray!50}\scalebox{1}{evaluation part}}\\\hline
     \multicolumn{6}{c}{\cellcolor{blue!50}} \\
     \multicolumn{6}{c}{\cellcolor{blue!50}\scalebox{1}{selection part}}\\
     \multicolumn{6}{c}{\cellcolor{blue!50}} \\\hline
     \end{array}\right)
\)
\caption{Parts of $B_m$.}
\label{fig:Br:portions}
\end{subfigure}
\caption{Parts  of  $B_r$.}
\label{fig:Bm}
\end{figure}

\medskip\noindent\emph{Matrices $B_m$.} 
Matrix $B_m$ is of dimension $(2c+1)\times (c\cdot 2^{\ell+1})$.  
For $j\in [c]$, in $i^{the}$ row, we have coefficients of $Y_{C_m}$ from 
$t(C_{m-1},i)$.    
In the $(c+1)^{st}$  row of $B_r$, we have coefficients of $Y_{C_m}$ from 
$f(C_m)$.
In the $(c+1+i)^{th}$ row of $B_m$, we have coefficients of $Y_r$ from $s(C_m,i)$.
That is $B_m$ is obtained by deleting the successor matching part from the construction of 
$B_r$ above. 
The entries of  $B_m$ are as follows, where $i\in \ZZ{c}$ and $a\in \ZZ{L}$.  
\begin{eqnarray}
&&B_m[{\ii}+1, {\ii}\cdot 2^{\ell+1}+2a+1]=B_m[{\ii}+1, {\ii}\cdot 2^{\ell+1}+2a+2]=L-1-a, \nonumber\\
&& B_m[c+1, {\ii}\cdot 2^{\ell+1}+2a+2]= 0, \text{ and } \nonumber\\
&& B_m[c+1, {\ii}\cdot 2^{\ell+1}+2a+2]=
\left\{ \begin{array}{ll}
1, & \mbox{if } \phi_{a}(X_{\ii}) \mbox{ satisfies } C_m,\\
0, & \mbox{otherwise.}
\end{array}\right. \nonumber\\
&& B_m[c+2+{\ii}, {\ii}\cdot 2^{\ell+1}+2a+1]=  B_m[c+2+{\ii}, {\ii}\cdot 2^{\ell+1}+2a+2]= 1, \label{eqn:Bm:selection} 
\end{eqnarray}

All other entries in $B_m$ are zeros. 
That is,  for all $\ii,\ip \in \ZZ{c}$ and $g\in [2^{\ell+1}]$ such that $\ii \neq \ip$, 
\begin{eqnarray}
&& B_m[1+{\ii}, {\ip}\cdot 2^{\ell+1}+g] =0 \label{eqn:Bm:zero1}\\
&& B_m[c+2+{\ii}, {\ip}\cdot 2^{\ell+1}+g] = 0, \label {eqn:Bm:zero3}\\
&& B_m[2c+2+\ii, {\ip}\cdot 2^{\ell+1}+g] = 0. \label{eqn:Bm:zero2}
\end{eqnarray}

%
%
%
%
%

\begin{figure}[t]
\scalefont{0.8}{
\begin{equation*}
 \left[ \begin{array}{cccccccccccccccc}
  \rowcolor{yellow!50}
0& 0& \color{red} 1& \color{red} 1& 2& 2& 3& 3& \color{gray} 0& \color{gray} 0& \color{gray} 0& \color{gray} 0& \color{gray} 0& \color{gray} 0& \color{gray} 0& \color{gray} 0 \\
\rowcolor{yellow!50}
\color{gray}  0& \color{gray} 0& \color{gray} 0& \color{gray} 0& \color{gray} 0& \color{gray} 0& \color{gray} 0& \color{gray} 0& 0& 0& 1& 1& 2& 2& \color{blue} 3& \color{blue} 3 \\
\rowcolor{gray!50}
1& 0& \color{red} 0& \color{red} 0& 1& 0& 1& 0& 0& 0& 1& 0& 0& 0& \color{blue} 0& \color{blue} 1 \\
\rowcolor{blue!50}
1& 1& \color{red}1& \color{red}1& 1& 1& 1& 1& 
\color{gray} 0& \color{gray} 0& \color{gray} 0& \color{gray} 0& \color{gray} 0& \color{gray} 0& \color{gray} 0& \color{gray} 0 \\ 
\rowcolor{blue!50}
\color{gray} 0& \color{gray} 0& \color{gray} 0& \color{gray} 0& \color{gray} 0& \color{gray} 0& \color{gray} 0& \color{gray} 0 &
1& 1& 1& 1& 1& 1& \color{blue} 1& \color{blue} 1 \\
\rowcolor{green!50}
3& 3& \color{red}2& \color{red}2& 1& 1& 0 & 0 & 
\color{gray} 0& \color{gray} 0& \color{gray} 0& \color{gray} 0& \color{gray} 0& \color{gray} 0& \color{gray} 0& \color{gray} 0 \\ 
\rowcolor{green!50}
\color{gray} 0& \color{gray} 0& \color{gray} 0& \color{gray} 0& \color{gray} 0& \color{gray} 0& \color{gray} 0& \color{gray} 0 &
3& 3& 2& 2& 1& 1& \color{blue} 0& \color{blue} 0 
\end{array} \right]
\end{equation*}
}

\scalefont{0.8}{
\begin{equation*}
\left[
\begin{array}{cccccccccccccccc}
\rowcolor{yellow!50}
0& 0& \color{red} 1& \color{red} 1& 2& 2& 3& 3& \color{gray} 0& \color{gray} 0& \color{gray} 0& \color{gray} 0& \color{gray} 0& \color{gray} 0& \color{gray} 0& \color{gray} 0 \\ 
\rowcolor{yellow!50}
\rowcolor{yellow!50}
\color{gray}  0& \color{gray} 0& \color{gray} 0& \color{gray} 0& \color{gray} 0& \color{gray} 0& \color{gray} 0& \color{gray} 0& 0& 0& 1& 1& 2& 2& \color{blue} 3& \color{blue} 3 \\
\rowcolor{gray!50}
1& 0& \color{red} 0& \color{red} 0& 1& 0& 1& 0& 0& 0& 1& 0& 0& 0& \color{blue} 1& \color{blue} 0 \\
\rowcolor{blue!50}
1& 1& \color{red}1& \color{red}1& 1& 1& 1& 1& 
\color{gray} 0& \color{gray} 0& \color{gray} 0& \color{gray} 0& \color{gray} 0& \color{gray} 0& \color{gray} 0& \color{gray} 0 \\ 
\rowcolor{blue!50}
\color{gray} 0& \color{gray} 0& \color{gray} 0& \color{gray} 0& \color{gray} 0& \color{gray} 0& \color{gray} 0& \color{gray} 0 &
1& 1& 1& 1& 1& 1& \color{blue} 1& \color{blue} 1 
\end{array}
\right]
\end{equation*}
}

\caption{Let $n=4,c=2,\ell=2$ and $C_r=x_1\vee \overline{x_2}\vee x_4$. The assignments are 
$\phi_0(X_0)=\{x_1=x_2=0\}, \phi_1(X_0)=\{x_1=0, x_2=1\}, \phi_2(X_0)=\{x_1=1, x_2=0\}, 
\phi_3(X_0)=\{x_1=x_2=1\}$, 
$\phi_0(X_1)=\{x_3=x_4=0\}, \phi_1(X_1)=\{x_3=0, x_4=1\}, \phi_2(X_1)=\{x_3=1, x_4=0\}, 
\phi_3(X_1)=\{x_3=x_4=1\}$. The entries defined according to $\phi_1(X_0)$ and $\phi_3(X_1)$ are colored 
red and blue respectively. If $1<r<m$, then the matrix on the left represents $B_r$ 
and if $r=1$, then $B_r$ can be obtained by deleting the yellow colored portion from the top matrix. 
The matrix on the right represents $B_m$. 
}
\label{fig:Br}
\end{figure}

\medskip\noindent\emph{Matrix $A$.}
Now we explain how the matrix $A$ is formed from $B_1,\ldots,B_m$. 
The matrices 
$B_1,\ldots,B_m$  are disjoint submatrices of $A$ and they  cover all non zero entries of $A$. Informally, the submatrices $B_1,\ldots,  B_m$ form  a chain such that 
  the rows corresponding 
to the successor matching part of $B_r$ will be the same as the rows in the   predecessor matching part  of 
$B_{r+1}$ (because of Equation~(\ref{eqn:consistancy}). A pictorial representation of $A$ can be found in Fig.~\ref{fig:A}. 
  Formally, 
let $I_1=[2c+1]$ 
and $I_m=[(m-1)(2c+1)+(c+1)]\setminus [(m-1)(2c+1)-c]$. 
For every $1<r<m$, let $I_r=[r(2c+1)]\setminus [(r-1)(2c+1)-c]$, 
and for  $r\in [m]$, let $J_r=
[r\cdot c\cdot 2^{\ell+1}]\setminus [(r-1)\cdot c\cdot 2^{\ell+1}]$. 
Now for each $r\in [m]$, the matrix $A[I_r,J_r]:=B_r$. 
All other entries of $A$ not belonging to any of  the submatrices  $A[I_r,J_r]$   
are  zero. 




Towards upper bounding the \pw\ of  $A$,  
we start with   some notations. 
We partition the set of columns of $A$ into $m$ parts $J_1,\ldots,J_m$ (we have already defined these sets) with  one part per clause.  
For each $r\in [m]$, $J_r$ is the set of columns associated with 
$Y_{C_r}$. We further divide $J_r$ into $c$ equal parts, one per variable set $Y_{C_r,\ii}$. 
These parts are  
$$P_{r,\ii}=\{(r-1)c\cdot 2^{\ell+1}+ \ii\cdot 2^{\ell+1}+1,\ldots, (r-1)c\cdot 2^{\ell+1}+ (\ii+1)\cdot 2^{\ell+1}\}, \,\ii\in \ZZ{c}.$$
In other words,  $P_{r,\ii}$ is the
set of columns corresponding to $Y_{C_r,\ii}$ and $\vert P_{r,\ii}\vert=2^{\ell+1}$. 
 We also put  $n'= m\cdot c \cdot 2^{\ell+1}$ to be  the number of columns in $A$. 


\begin{lemma} 
\label{lemma:pathwidthbounded}
The \pw of the column matroid of $A$ is at most $c+4$
\end{lemma}
\begin{proof}
Recall that $n'=m\cdot c\cdot 2^{\ell+1}$, be the number of columns in $A$ and 
$m'$ be the number of rows in $A$.
To prove that the \pw of $A$ is at most $c+4$, it is sufficient to show that 
 for all $j\in [n'-1]$, 
 \begin{equation}
 \operatorname{dim}\langle  \operatorname{span}(A|\{1,\dots, j\})\cap\operatorname{span}(A|\{j+1, \dots, n'\}) \rangle \leq c+3. \label{eqn:dim:intersection}
 \end{equation}
 
The idea for proving Equation~\eqref{eqn:dim:intersection} is based on the following observation.  
For $V'=A|\{1,\dots, j\}$ and $V''=A|\{j+1, \dots, n'\}$, let 
$$I=\{\q\in [m']~|~\mbox{there exists } v'\in V'  \text{ and } v''\in V'' \text{ such that }   v'[\q]\neq v''[\q] \neq 0\}.$$ 
Then the  dimension of  $\operatorname{span}(V')\cap \operatorname{span} (V'')$ is at most $\vert I \vert$. 
Thus to prove  \eqref{eqn:dim:intersection}, for each $j\in [n'-1]$, we construct the corresponding set $I$ and show that 
its cardinality is at most $c+3$. 

We proceed with the details. 
Let $v_1,v_2,\ldots,v_{n'}$ be the column vectors of $A$. 
Let $j\in [n'-1]$. 
Let   $V_1=\{v_1,\ldots,v_j\}$ and 
$V_2=\{v_{j+1},\ldots,v_{n'}\}$. We need to show that 
$\operatorname{dim}\langle  \operatorname{span}(V_1)\cap\operatorname{span}(V_2) \rangle \leq c+3$.
Let $$I'=\{\q \in [m']~|~\mbox{there exists $v\in V_1$ and $v'\in V_2$ such that   $v[\q]\neq 0\neq v'[\q]$}\}.$$ 
We know that $[n']$ is partitioned into parts $P_{r',\ii'}, r'\in [m], \ii'\in \ZZ{c}$. 

\begin{center}
\textbf{Fix $r\in [m]$ and $\ii\in \ZZ{c}$ such that $j \in P_{r,\ii}$}.
\end{center}
 Let 
$j=(r-1)c\cdot 2^{\ell+1}+\ii\cdot 2^{\ell+1}+\bb$, where $\bb \in [2^{\ell+1}]$.
Let $\q_1=\max\{0, (r-1)(2c+1)-c\}$,  $\q_2=r(2c+1)$,  $j_1=(r-1)\cdot c\cdot 2^{\ell+1}$, and $j_2=r\cdot c\cdot 2^{\ell+1}$
Then 
$[\q_2]\setminus [\q_1]=I_r$ and $[j_2]\setminus [j_1]=J_r$ (recall the definition of sets $I_r$ and $J_r$). 

By the decomposition of matrix  $A$, for every $\q>\q_2$ and for every vector $v\in V_1$, we have 
$v[q]=0$.  Also, for every $\q\leq \q_1$ and for any $v\in V_2$, we have that 
$v[\q]=0$. This implies that  $I' \subseteq [\q_2]\setminus [\q_1]=I_r$.  
Now we partition $I_r$ into $4$ parts: $R_1,R, S$, and $R_2$, These parts are defined as follows.  
\begin{eqnarray}
R_1&=& 
\left\{ \begin{array}{ll}
\emptyset, & \mbox{if } r=1, \\
\{(r-2)(2c+1)+\ip~|~\ip\in \ZZ{c}\}, & \mbox{otherwise,}
\end{array}\right. \nonumber \\
R&=&\{(r-1)(2c+1)+1\},  \label{Eqn:R}\\
S&=& 
\{(r-1)(2c+1)+2+ \ip ~|~\ip\in \ZZ{c}]\}
 \nonumber \\
R_2&=& 
\left\{ \begin{array}{ll}
\emptyset, & \mbox{if } r=m, \\
\{(r-1)(2c+1)+c+2+ \ip ~|~\ip \in \ZZ{c}\} ,& \mbox{otherwise } 
\end{array}\right. \nonumber 
\end{eqnarray}

\begin{claim}
\label{claim:outofwindow0}
For each $r'\in [m], \q\notin I_{r'}$ and $j''\in J_{r'}, v_{j''}[\q]=0$. 
\end{claim}
\begin{proof}
The non-zero entries in $A$ are covered by the disjoint sub-matrices 
$A[I_{r'},J_{r'}]=B_{r'}, {r'}\in [m]$. Hence the claim follows.  
\end{proof}

\begin{claim}
$\vert I'\cap R_1 \vert \leq c-(\ii-1)$.
\label{claim:R1}
\end{claim}
\begin{proof}
When $r=1$, $R_1=\emptyset$ and the claim trivially follows. 
Let $r>1$, and let  $\q \in R_1$ be such that $\q <(r-2)(2c+1)+\ii$. Then   $q=(r-2)(2c+1)+1+\ip$ for some 
  $0 \leq  \ip< \ii$. 
Notice that $\q \notin I_{r'}$ for every $r'>r$. By Claim~\ref{claim:outofwindow0},  for every $v\in \bigcup_{r'>r} J_{r'}$, $v[\q]=0$. 
Now consider the vector $v_{j''}\in V_2 \setminus  (\bigcup_{r'>r} J_r')$. 
Notice that $j''>j$ and $j''\in J_r$.  Let $j''=j+a=(r-1)c\cdot 2^{\ell+1}+\ii \cdot 2^{\ell+1}+g+a$ 
for some $a\in [rc2^{\ell+1}-j]$.  
 From the decomposition of $A$, 
$v_{j''}[\q]=B_r[\ip+1, \ii \cdot 2^{\ell+1}+\bb+a]= 0$, by (\ref{eqn:Br:zero1}).  Thus
for every $\q \in R$, $\q<(r-2)(2c+1)+\ii$ and $v\in V_2$, $v[\q]=0$.

This implies that 
$$\vert I'\cap R_1\vert \leq \vert \{\q \geq (r-2)(2c+1)+\ii\}\cap R_1 \vert \leq c-(\ii-1).\qedhere$$
\end{proof}



\begin{claim}
$\vert I'\cap R_2\vert \leq \ii$. 
\label{claim:R2}
\end{claim}

\begin{proof}
When $r=m$, $R_2=\emptyset$ and the claim trivially holds. 
So, now let $r<m$ and consider any $\q \in R_2 \cap \{ \q' >(r-1)(2c+1)+c+2+\ii\}$.  
Let $\ip>\ii$ such that $\q=(r-1)(2c+1)+c+2+\ip$. 
Notice that $\q\notin I_{r'}$ for any $r'<r$. Hence, by Claim~\ref{claim:outofwindow0},  for any $v\in \bigcup_{r'<r} J_{r'}$, $v[\q]=0$. 
Now consider any vector $v_{j''}\in V_1 \setminus  (\bigcup_{r'<r} J_r')$. 
Notice that $j''\leq j$ and $j''\in J_r$.  Let $j''=(r-1)c\cdot 2^{\ell+1}+i'' \cdot 2^{\ell+1}+a$  
for some $a\in [2^{\ell+1}]$ and $i''\leq i < \ip$.  
From the decomposition of $A$, 
$v_{j''}[\q]=B_r[2c+2+\ip, i'' \cdot 2^{\ell+1}+a]= 0$, by 
(\ref{eqn:Br:zero2}).  Hence we have shown that 
for any $\q\in R$, $\q>(r-2)(2c+1)+c+2+\ii$ and $v\in V_1$, $v[\q]=0$. This implies that 
$$\vert I'\cap R_2\vert \leq \vert \{\q \leq (r-1)(2c+1)+c+2+\ii\}\cap R_1 \vert \leq \ii.\qedhere$$
\end{proof}
\begin{claim}
$\vert I'\cap S \vert \leq 1$. 
\label{claim:R2'}
\end{claim}
\begin{proof}
%

Consider any $\q \in S$.  
Let $\ip \in \ZZ{c}$ such that $\q=(r-1)(2c+1)+2+\ip$. 
Notice that $\q \notin I_{r'}$ for any $r'<r$, and hence, by Claim~\ref{claim:outofwindow0},  for any $v\in \bigcup_{r'<r} J_{r'}$, $v[\q]=0$. Also notice that $\q \notin I_{r'}$ for any $r'>r$, 
and hence, by Claim~\ref{claim:outofwindow0},  for any $v\in \bigcup_{r'>r+1} J_{r'}$, $v[\q]=0$. 
So the only potential $j''$ for which $v_{j''}[\q]\neq 0$,  are from $J_r$. 

We claim that if  
$\q \in I' \cap S$, then $\q =(r-1)(2c+1)+2+\ii$. 
Suppose $\q \in I' \cap S$ and  $\q <(r-1)(2c+1)+2+\ii$. 
Let  $\q=(r-1)(2c+1)+2+\ip$, where $0\leq \ip <\ii$. 
Then by the decomposition of $A$, for any $j''>j$, 
$v_{j''}[\q]=B_r[c+2+\ip,j''-(r-1)c2^{\ell+1}]=B_r[c+2+\ip, i_1 2^{\ell+1}+a]$, where $c-1\geq i_1\geq \ii$ and $a\in [2^{\ell+1}]$. 
Thus by (\ref{eqn:Br:zeroX}),   $v_{j''}[\q]=B_r[c+2+\ip,i_1 2^{\ell+1}+a]=0$. 
This 
contradicts the assumption that $\q \in I' \cap S$. 

Suppose $\q \in I' \cap S$ and  $\q >(r-1)(2c+1)+c+2+\ii$. 
Let  $\q=(r-1)(2c+1)+c+2+\ip$, where $\ii<\ip<c$. 
Then by the decomposition of $A$, for any $j''\leq j$, 
$v_{j''}[\q]=B_r[c+2+\ip,j''-(r-1)c 2^{\ell+1}]=B_r[c+2+\ip, i_1 2^{\ell+1}+a]$, where $0\leq i_1\leq \ii$, $a\in [2^{\ell+1}]$.  
Thus by (\ref{eqn:Br:zeroX}),  $v_{j''}[\q]=B_r[c+2+\ip, i_1 2^{\ell+1}+a]=0$. This 
contradicts the assumption that $i\in I' \cap S$. 
This implies that  $\vert I'\cap S \vert \leq 1$.
This completes the proof of the claim. 
\end{proof}
Therefore, we have 
\begin{eqnarray*}
\vert I' \vert &=& \vert I' \cap I_r \vert 
\qquad\qquad\qquad\qquad\qquad\qquad\qquad\qquad\qquad\qquad\qquad (\mbox{Because } I'\subseteq I_r)\\
&=&\vert I' \cap R_1 \vert  + \vert I' \cap R \vert + \vert I' \cap S \vert + \vert I' \cap R_2 \vert \qquad \qquad\qquad (\mbox{By ~\eqref{Eqn:R}})\\
&\leq& c-(\ii-1)+1+1+\ii    \qquad \qquad\qquad \qquad(\mbox{By Claims~\ref{claim:R1},\ref{claim:R2} and \ref{claim:R2'}})\\ 
&=&c+3 
\end{eqnarray*}
This completes the proof of the lemma.
\end{proof}
%
\begin{proof}[Proof of Theorem~\ref{thm:lowbranchwidth}.] 
We prove the theorem by assuming a fast algorithm for 
\IP and use it to give a fast algorithm for  {\sc CNF-SAT}, refuting SETH. 
Let $\psi$ be an instance of {\sc CNF-SAT} 
with $n_1$ variables and $m_1$ clauses.  We choose a sufficiently large constant $c$ such that 
$(1-\epsilon)+\frac{4(1-\epsilon)}{c}+\frac{a}{c}<1$ holds. 
We use the reduction mentioned in Lemma~\ref{lemtechnicalintro} and construct an instance  $\cm x= \tv, x\geq 0$, of \IP which has a solution if and only if $\psi$ is satisfiable. 
The reduction takes time $\OO(m_1^2 2^{\frac{n_1}{c}})$.  
Let $\ell=\lceil \frac{n_1}{c}\rceil$.  
The constraint matrix $\cm$ has dimension  $((m_1-1)(2c+1)+1+c)\times (m_1\cdot c\cdot 2^{\ell+1})$ 
and the largest entry in   vector $\tv$ does not exceed $2^{\ell}-1$. The \pw of $M(\cm)$ is at most $c+4$.

Assuming that any instance of \IP with non-negative constraint matrix of \pw $k$ is solvable in time $f(k)(\valueB+1)^{(1-\epsilon)k}(mn)^a$, 
where $d$ is the maximum value in an entry of $b$ and $\epsilon,a>0$ are constants,  we have that 
$\cm x= \tv, x\geq 0$,  is solvable in time 
\[
f(c+4) \cdot 2^{\ell \cdot (1-\epsilon)(c+4)} \cdot2^{\ell\cdot a} \cdot m_1^{\OO(1)}= 2^{\frac{n_1}{c}(1-\epsilon)(c+4)} \cdot
2^{\frac{n_1\cdot a}{c}} \cdot m_1^{\OO(1)} = 2^{n_1\left( (1-\epsilon)+\frac{4(1-\epsilon)}{c} +\frac{a}{c} \right)} \cdot  m_1^{\OO(1)}.
\]
Here the constant $f(c+4)$ is subsumed by the term $m_1^{\OO(1)}$. 
Hence the total running time for testing whether $\psi$ is satisfiable or not, is, 
\[
\OO(m_1^2 2^{\frac{n_1}{c}})+2^{n_1\left( (1-\epsilon)+\frac{4(1-\epsilon)}{c} +\frac{a}{c} \right)} m_1^{\OO(1)} 
= 2^{n_1\left( (1-\epsilon)+\frac{4(1-\epsilon)}{c} +\frac{a}{c} \right)} m_1^{\OO(1)}= 2^{\epsilon'\cdot n_1} m_1^{\OO(1)},  
\]
where $\epsilon'=(1-\epsilon)+\frac{4(1-\epsilon)}{c}+\frac{a}{c}<1$. 
This completes the proof of Theorem~\ref{thm:lowbranchwidth}. 
%
%
\end{proof}

\subsection{Proof Sketch of Theorem~\ref{thm:lowentries}}\label{sec:lowentries}

In this section we prove  Theorem~\ref{thm:lowentries}:  \IP{} with non-negative matrix $A$ cannot be solved in time $f(\valueB)(\valueB+1)^{(1-\epsilon)k}(mn)^{\OO(1)}$ for any function 
$f$ and $\epsilon>0$, unless SETH fails,  where 
$k$ is the \pw of the column matroid of $A$.

In Section~\ref{sec:lowbranchwidth}, we gave a reduction from {\sc CNF-SAT} 
to \IP. However in this reduction the values in the constraint matrix  $\cm$
and target  vector $\tv$ can be as  large as  $2^{\lceil \frac{n}{c}\rceil}-1$, 
where $n$ is the number of variables in the {\sc CNF}-formula $\psi$ and $c$ is 
a constant. Let $m$ be the number of clauses in $\psi$. In this section 
we briefly explain how to get rid of these large values, at the cost 
of making large, but still bounded \pw.   
From a {\sc CNF}-formula $\psi$, we construct a matrix $A=\cm$ 
as described in Section~\ref{sec:lowbranchwidth}. 
The only rows 
in $A$ which contain values strictly greater than $1$ (values other than $0$ or $1$)
are the ones corresponding to the constraints defined by Equation~(\ref{eqn:consistancy}). 
 %
In other words,  the values greater than $1$ are in the rows in 
yellow/green colored portion in Figure~\ref{fig:A}. 
Recall that $\ell=\lceil \frac{n}{c}\rceil$ and the largest value in $A$  is $2^{\ell}-1$.
Any number less than  or equal to $2^{\ell}-1$ can be represented by a
binary string of length $\ell=\frac{n}{c}$. 
Now we rewrite the Equation~(\ref{eqn:consistancy}), by 
$\ell$ new equations. For each $j\in [\ell]$ and $N\in{\mathbb N}$, let $b_j(N)$, represents the $j^{th}$ bit 
in the $\ell$-bit binary representation of $N$. Then for each for all $r\in [m-1]$, $i\in\ZZ{c}$ and $j\in [\ell]$, we have a system of constraints   
\begin{eqnarray}
\sum_{\substack{a \in \ZZ{2L} }} \left(b_j\left(\lfloor \frac{a}{2}\rfloor\right)\cdot y_{C_r,i,a}\right)   + \left(b_j(L-1-\lfloor \frac{a}{2}\rfloor)\cdot y_{C_{r+1},i,a}\right) &=&1 
\end{eqnarray}

In other words, let $P=\{(r-1)(2c+1)+c+1+\ii~|~r\in[m-2], \ii\in \ZZ{c}\}$.
The rows of $A$ containing values larger than one are indexed by $P$. 
Now we construct a new matrix 
$A'$ from $A$ by replacing each row of $A$ whose index is in the set $P$ with $\ell$ rows and for any value 
$A[i,j], i\in P$ we write its $\ell$-bit binary representation in the column corresponding 
to $j$ and the newly added $\ell$ rows of $A'$. 
That is,  for any $\gamma\in P$, we replace the row $\gamma$ with $\ell$ rows,  $\gamma_{1},\ldots,\gamma_{\ell}$. Where, for any 
$j$, if the value $A[\gamma,j]=W$ then $A'[\gamma_{k},j]=\eta_k$, where $\eta_k$ is the $k^{th}$ bit in the 
$\ell$-sized binary representation of $W$.  

Let $m'$ be the number of rows in 
$A'$. Now the target vector $b'$ is defined as $b'[i]=1$ for all $i\in [m']$. 
This completes the construction of the reduced \IP{} instance $A'x=b'$.  
The correctness proof of this reduction is using arguments similar to those used for the correctness of  
Lemma~\ref{lemma:brachwidthcorrect}.    

\begin{lemma} 
\label{lemma:pathwidthvaluesmall}
The \pw of the column matroid of $A'$ is at most $(c+1)\frac{n}{c}+3$.
\end{lemma}
\begin{proof}
We sketch the proof, which is similar to the proof of Lemma~\ref{lemma:pathwidthbounded}. 
We  define $I'_r$ and $J'_r$ for any $r\in [m]$ like $I_r$ and $J_r$  in Section~\ref{sec:lowbranchwidth}. 
In fact, the rows in $I'_r$ are the rows obtained from $I_r$ in the process explained above to construct $A'$ from $A$. 
We need to show that 
 $\operatorname{dim}\langle  \operatorname{span}(A'|\{1,\dots, j\})\cap\operatorname{span}(A'|\{j+1, \dots, n'\}) \rangle \leq (c+1)\frac{n}{c}+2$ for all $j\in [n'-1]$, where $n'$ is the number of columns in $A'$.
The proof proceeds by bounding the number of indices $I$ 
such that for any $\q \in I$ there exist vectors    $v\in A'|\{1,\dots, j\}$  and $u\in A'|\{j+1, \dots, n'\}$ with  $v[\q]\neq 0 \neq u[\q]$. 
By arguments similar to the ones used in the proof of Lemma~\ref{lemma:pathwidthbounded}, 
we can show that for any $j\in [n'-1]$, the corresponding set $I'$ of indices is a subset of $I'_r$ for some $r\in [m]$.  
Recall the partition of $I_r$ into $R_1,R,S$ and $R_2$ in Lemma~\ref{lemma:pathwidthbounded}.  
We partition 
$I'_r$ into parts $Q_1,W,U$ and $Q_2$. 
Notice that $R_1,R_2\subseteq P$, where $P$ is the set of rows which covers all 
 values strictly greater than $1$.   
The set $Q_1$ and $Q_2$ are obtained from $R_1$ and $R_2$, respectively, by the process 
mentioned above to construct $A'$ from $A$. 
That is, each row in $R_i, i\in \{1,2\}$ is replaced by $\ell$ rows in $Q_i$. Rows in 
$W$ corresponds to  rows in $R$ and $U$ corresponds to the rows in $W$. 
This allows us to bound the following terms for some $\ii \in \ZZ{c}$:
\begin{eqnarray*}  
\vert I'\cap Q_1\vert &\leq& (c-(\ii -1))\ell=(c-(\ii-1))\ell,\\
\vert I'\cap Q_2\vert &\leq& \ii \cdot \ell, \\
\vert I'\cap U\vert &\leq& 1, \text{ and}\\
\vert I'\cap W\vert &\leq& 1.
\end{eqnarray*}
By using the fact that $I'\subseteq I'_r$ and the above system of inequalities, 
we can show  that 
$$\operatorname{dim}\langle  \operatorname{span}(A'|\{1,\dots, j\})\cap\operatorname{span}(A'|\{j+1, \dots, n'\}) \rangle \leq (c+1)\lceil\frac{n}{c}\rceil+2.$$ 
This completes the proof sketch of the lemma.    
\end{proof}
Now the proof of the theorem follows from Lemma~\ref{lemma:pathwidthvaluesmall} and the correctness 
of the reduction (it is similar to the arguments in the proof of Theorem~\ref{thm:lowbranchwidth}). 
 \section{Proof of Theorem~\ref{thmCGlin}}\label{sec:thmCGlin}
In this section, we sketch how the proof  of Cunningham and Geelen \cite{CunninghamG07}
of  Theorem~\ref{thmCG}, can be adapted to prove  Theorem~\ref{thmCGlin}. 
Recall that a path decomposition of width $k$ can be 
obtained  in $f(k)\cdot n^{\OO(1)}$ time for some function $f$ by making use of the algorithm   by Jeong et al.~\cite{Jeong0O16}. However, we do not know if such a path decomposition can be constructed in  time  $\OO ((\valueB+1)^{k+1}) n^{\OO(1)}$, so the assumption that a path decomposition is given is essential.

Roughly speaking, the only difference in the proof is that when parameterized by the branch-width, the most time-consuming operation is the ``merge" operation, when we have to construct a new set of partial solutions with at most 
$(\valueB+1)^k$ vectors from two already computed sets of sizes $(\valueB+1)^k$ each. Thus to construct a new set of vectors, one has to go through all possible pairs of vectors from both sets, which takes time roughly $(\valueB+1)^{2k}$. 
 For \pw parameterization,   the new partial solution set  is constructed  from two sets, but this time one set contains at most  $(\valueB+1)^k$ vectors while the second contains at most  $\valueB+1$ vectors. This allows us to construct the new set in time roughly $(\valueB+1)^{k+1}$.
 
 \medskip

Recall that for $X\subseteq [n]$, we define 
$S(A,X)=\operatorname{span}(A|X)\cap\operatorname{span}(A|E \setminus X)$, 
where $E=[n]$.
The key lemma in the proof of Theorem~\ref{thmCG} is the following. 
\begin{lemma}[\cite{CunninghamG07}]
\label{lemma:CGboundedpartialsoln}
Let $A\in \{0,1,\ldots,\valueB\}^{m\times n}$ and $X\subseteq [n]$ such that $\lambda_{M(A)}(X)=k$. 
Then the number of vectors in  $S(A,X)\cap \{0,\ldots,\valueB\}^m$ is at most $(\valueB+1)^{k-1}$. 
\end{lemma}

To prove  Theorem~\ref{thmCGlin},  without loss of generality, we  
assume that the columns of $A$ are ordered in such a way that 
for every $j\in [n-1]$,  
$$\operatorname{dim}\langle  \operatorname{span}(A|\{1,\dots, i\})\cap\operatorname{span}(A|\{i+1, \dots, n\}) \rangle \leq k-1.$$ 

Let $A'=[A,b]$. That is $A'$ is obtained by appending the column-vector $b$ to  the end of $A$.  
Then for each $i\in [n]$, 
\begin{equation}\label{eq_dim}
\operatorname{dim}\langle  \operatorname{span}(A'|\{1,\dots, i\})\cap\operatorname{span}(A'|\{i+1, \dots, n+1\}) \rangle \leq k.
\end{equation}
Now we use  dynamic programming to check  whether the following conditions are satisfied. 
For  $X \subseteq [n+1]$, let $\BB(X)$ be  the set of all vectors $b'\in \Z^m$ such that 
\begin{itemize}
\item[(1)] $0 \leq b' \leq b$, 
\item[(2)] there exists $z\in \Z^{\vert X \vert}$  such that $(A'|X)z=b'$, and 
\item[(3)] $b'\in  S(A', X)$.
\end{itemize}
Then \IP has a solution if and only if $b\in \BB([n])$. Initially the algorithm computes for all $i\in [n]$, $\BB(\{i\})$ and by Lemma~\ref{lemma:CGboundedpartialsoln}, 
we have that  $\vert \BB(\{i\})\vert \leq \valueB+1$. In fact $\BB(\{i\})\subseteq \{a\cdot v~|~v \mbox{ is the $i^{th}$ column vector of  $A'$ and }a\in [\valueB+1] \}$. Then for each $j\in [2,\ldots n]$ the algorithm computes $\BB([j])$ in   increasing order of $j$ 
and outputs {\sc Yes} if and only if  $b\in \BB([n])$. 
That is, $\BB([j])$ is computed from the  already computed sets $\BB([j-1])$ and $\BB(\{j\})$. 
Notice that $b'\in \BB([j])$ if and only if 
\begin{itemize}
\item[(a)] there exist $b_1\in \BB(\{1,\ldots,j-1\})$ and $b_2\in \BB(\{j\})$ such that $b'=b_1+b_2$,
\item[(b)] $b'\leq b$ and 
\item[(c)] $b'\in S(A',[j])$. 
\end{itemize}  
 So the algorithm enumerates vectors $b'$ satisfying condition $(a)$, and each such vector $b'$ is included in $\BB([j])$,  
if  $b'$ satisfy conditions $(b)$ and $(c)$. Since by \eqref{eq_dim} and Lemma~\ref{lemma:CGboundedpartialsoln},  $\vert \BB([j-1])\vert \leq (\valueB+1)^k$ and $\vert \BB(\{j\})\vert\leq \valueB+1$, 
the number of vectors satisfying condition $(a)$ is $(\valueB+1)^{k}$, and hence the exponential factor of the 
required running time  follows.
This provides the bound on the claimed exponential dependence in the running time of the algorithm. The bound on the polynomial component of the running time follows from exactly the same arguments as in~\cite{CunninghamG07}.

\section{Conclusion}\label{sec:concl}

%
%

In a previous version of this paper on ArXiv~\cite{FominPRS16}  we pointed out that it was unknown whether the  algorithm of  Papadimitriou~\cite{Papadimitriou81}  is asymptotically optimal. This question has now been answered by  Eisenbrand and Weismantel in~\cite{EisenbrandW18}, who gave an algorithm solving  
\IP with  an  $m\times n$ matrix $A$  in time $(m\cdot \valueA)^{\OO(m)}\cdot \valueB^2$, where $\valueA$ is the upper bound on the absolute values of the entries of $A$. 
While 
Theorems~\ref{thm:ETHIP} and ~\ref{thm:ETHIP2} come close to this bound,
 the precise multivariate complexity of \IP with respect to the parameters $n$, $m$, $\valueA$,  and $\valueB$ is not fully clear and our work leaves   some unanswered questions regarding the landscape of   tradeoffs between the parameters. For instance, 
is it possible to solve  \IP in time

\begin{itemize}
\item $ (m\cdot n\cdot \valueA)^{o(m)}\cdot (\valueB)^{\OO(1)}$, or 
\item $ (m\cdot n\cdot \valueA\cdot \valueB)^{o(m)} $?
\end{itemize} 

%

\medskip
Or could one improve our lower bound results to rule out such algorithms? 
While our SETH-based lower bounds 
for \IP with non-negative constraint matrix are tight for \pw parameterization, there is a ``$(\valueB+1)^k$ to $(\valueB+1)^{2k}$ gap'' between lower and upper bounds for branch-width parameterization.  
Closing this gap is the first natural question. 




The proof of Theorem~\ref{thmCG} given by  Cunningham and Geelen consists of two parts.  
 The first part bounds the number of potential partial solutions 
corresponding to any edge of the branch decomposition tree by 
$(\valueB+1)^k$.  The second part is the dynamic programming 
over the branch decomposition using the fact that  the number of 
potential partial solutions is bounded. The bottleneck in  the algorithm of Cunningham and Geelen   is the following subproblem. We are given two vector sets $A$ and $B$ of partial solutions, each set of size at most $(\valueB+1)^k$.  We need to construct a new vector set $C$ of partial solutions, where the set $C$ will have size  at most 
$(\valueB+1)^k$ and each vector from $C$ is the \emph{sum} of a vector from $A$ and a vector from $B$. 
 Thus to construct the new set of vectors, one has to go through all possible pairs of vectors from both sets $A$ and $B$, which takes time roughly $(\valueB+1)^{2k}$.  

A tempting approach towards improving the running time of this particular step could be the use of {\em fast subset convolution} 
or {\em matrix multiplication} tricks, which work very well for ``join'' operations in dynamic programming algorithms over tree and branch decompositions 
of graphs~\cite{Dorn06,RBR09,cut-and-count}, see also \cite[Chapter~11]{cygan2015parameterized}.  
Unfortunately, we have reason to suspect that these tricks may \emph{not} help for matrices:  solving the above subproblem in time $(\valueB+1)^{(1-\epsilon)2k}n^{\OO(1)}$ for any $\epsilon>0$ 
would imply that $3$-SUM is  solvable in time $n^{2-\epsilon}$, which is believed to be unlikely.
(The $3$-SUM problem asks whether a given set of $n$ integers contains three elements that sum to zero.) 
Indeed, consider an equivalent version of $3$-SUM,  named $3$-SUM$^\prime$,   which  is defined as follows. 
Given $3$ sets of integers $A,B$ and $C$ each of cardinality $n$, and the objective is to check whether there exist $a\in A$, $b\in B$ and $c\in C$ 
such that $a+b=c$. Then,  $3$-SUM is  solvable in time $n^{2-\epsilon}$ if and only if $3$-SUM$^\prime$ is as well (see Theorem~$3.1$ in~\cite{GajentaanO95}).
However, the problem $3$-SUM$^\prime$  is equivalent to the most time consuming step in the algorithm of Theorem~\ref{thmCG}, 
where the integers in the input of $3$-SUM$^\prime$ can be thought of as length-one vectors. 
While this observation does not \emph{rule out} the existence of an algorithm solving 
(IP) with constraint matrices of branch-width $k$ in time
$(\valueB+1)^{(1-\epsilon)2k}n^{\OO(1)}$, it indicates that any interesting improvement in the running time would  require a completely different approach.


  \newpage
 \bibliography{book_pc.bib}

\end{document}